\documentclass[11pt,a4paper]{amsart}
\usepackage[left=3cm,bottom=3cm,right=3cm,top=3cm]{geometry}
\usepackage{amsmath} 
\usepackage{nccmath}  
\usepackage{graphicx}
\usepackage[margin=10pt,font=small,labelfont=bf,labelsep=endash]{caption}
\usepackage{mathrsfs}
\usepackage[utf8]{inputenc}
\usepackage{amssymb}
\usepackage{datetime}
\usepackage{color}
\numberwithin{equation}{section}
\newtheorem{defi}{Definition}[section]
\newtheorem{lemma}{Lemma}[section]
\newtheorem{cor}{Corollary}[section]
\newtheorem{prop}{Proposition}[section]
\newtheorem{theorem}[lemma]{Theorem}
\newtheorem{rem}[lemma]{Remark}

\begin{document}
\title[Rotating Clouds of Charged Vlasov Matter in GR]{Rotating Clouds of Charged Vlasov Matter in General Relativity}
\author{Maximilian Thaller}
\date{\today}

\begin{abstract}
The existence of stationary solutions of the Einstein-Vlasov-Maxwell system which are axially symmetric but not spherically symmetric is proven by means of the implicit function theorem on Banach spaces. The proof relies on the methods of \cite{akr14} where a similar result is obtained for uncharged particles. Among the solutions constructed in this article there are rotating and non-rotating ones. Static solutions exhibit an electric but no magnetic field. In the case of rotating solutions, in addition to the electric field, a purely poloidal magnetic field is induced by the particle current. The existence of toroidal components of the magnetic field turns out to be not possible in this setting.
\end{abstract}

\maketitle

\tableofcontents

\newpage 

\section{Introduction}
The Einstein-Vlasov-Maxwell system (EVM-system) describes an ensemble of charged particles whose motion is governed by gravity and an electro-magnetic field but which do not interact via collisions. In the framework of General Relativity gravity is described by the curvature of the manifold, the space-time, on which the particles live. Both the space-time curvature and the electro-magnetic field are generated collectively by the particles themselves. In contrast to the Einstein-Vlasov system, which only takes into account gravity, particles described by the EVM-system are not freely falling, i.e.~their trajectories are not geodesics. \par
In this article the existence of stationary, rotating solutions of the EVM-system is proven by means of the implicit function theorem. The proof is a generalisation of \cite{akr14}, where the existence of rotating, stationary solutions of the Einstein-Vlasov system with uncharged particles is proved, to the case where the particles are charged and hence induce an electro-magnetic field. In the context of kinetic theory this method has already been used in \cite{r00} to show the existence of stationary, rotating solutions of the Vlasov-Poisson system. The idea of this method is to introduce a parameter $\lambda$ to the system which can ``turn on'' rotation and to perturb the system around a spherically symmetric, static solution without rotation. To this end one considers a functional $\mathfrak F : \mathcal X \times [-\delta,\delta]\to \mathcal X$, where $\mathcal X$ is a suitable function space which will contain the solution and $[-\delta, \delta]$ is the interval in which the parameter $\lambda$ will lie. The operator is constructed such that if $\mathfrak F(\zeta,\lambda) = 0$ then $\zeta$ is a collection of functions which constitute a solution of the Vlasov-Poisson system with the parameter $\lambda$. The solution $\zeta_0$, corresponding to $\lambda=0$, is known and we have $\mathfrak F(\zeta_0,0)=0$. The main part of the work consists in showing that the implicit function theorem can be applied. Then it follows that to each $\lambda\in (-\delta,\delta)$ there exists $\zeta_\lambda \in \mathcal X$ such that $\mathfrak F(\zeta_\lambda, \lambda)=0$. This collection $\zeta_\lambda$ of functions consequently solves the Vlasov-Poisson system and this solution is axially symmetric but not spherically symmetric. It is in the nature of this method that the obtained rotating solutions have small overall angular momentum.\par
In \cite{akr11} a similar method with a different set up has been used to show the existence of axially but not spherically symmetric, static solutions of the Einstein-Vlasov system. In this context it was used that the Vlasov-Poisson system is the non-relativistic limit of the Einstein-Vlasov system, in the sense that a solution of the Einstein-Vlasov system converges to a solution of the Vlasov-Poisson system if the speed of light $c$ goes to infinity. So besides $\lambda$, the speed of light $c$ has been introduced to the system as a second parameter. Perturbing off a spherically symmetric, static solution of the Vlasov-Poisson system in those two parameters $\lambda$ and $c$ yields an axially but not spherically symmetric, static solution of the Einstein-Vlasov system. The deviation from spherical symmetry is small but by a scaling argument the solution can be made fully relativistic, i.e.~$c=1$. In \cite{akr14} further technical insights made it possible to include rotation into the picture. \par
Lichtenstein developed a method based on the implicit function theorem to construct rotating fluid bodies \cite{l18, l33} in Newtonian gravity. This approach has later been reformulated in a modern mathematical language \cite{h94} and improved \cite{h95}. In \cite{abs08, abs09} the authors use an implicit function argument to construct axially symmetric static and rotating elastic bodies in Einstein gravity. In a series of papers of which the last one is \cite{cdkkmr18} the authors construct stationary solutions of the Einstein equations with negative cosmological constant without any symmetries. Many different matter models can be included, such as a scalar field, Maxwell, or Yang-Mills. \par 
Space-times with rotating, charged matter configurations  have been studied in the literature by analytical and numerical means, see e.g.~\cite{bbgn95, cpl00, fr12}. An important motivation for these studies is the modelling of rotating stars or neutron stars with a magnetic field. In these articles the matter is modelled as a perfect fluid and different shapes of the magnetic field can be observed depending on the assumptions on the fluid, like an equation of state or conductivity properties. For example rotating solutions with no poloidal magnetic field can be constructed, cf.~\cite{fr12}. These works can serve as a source of intuition for the study of rotating clouds of Vlasov matter. There is however an important difference. When studying a perfect fluid, the Einstein-Euler system (which describes a space-time containing matter of the type of a perfect fluid) has to be supplemented by an equation of state which captures the physical properties of the fluid under consideration. Depending on the choice of the equation of state, different matter configurations and different electro-magnetic fields can be constructed. For Vlasov matter however there is much less variety in the physical properties of the solutions that can be obtained. The basic assumptions on the particles' behaviour and how the energy and the angular momentum is distributed among the particles (this is sometimes referred to as a {\em microscopic equation of state}) already determines the macroscopic character of the solutions. It turns out that rotating solutions of the EVM-system must have a poloidal magnetic field but no toroidal magnetic field. \par
We briefly mention that in the non-relativistic setting a variety of different axially symmetric solutions can be constructed explicitly, cf.~for example \cite{galactic_dynamics}. A well studied class of these solutions are disk solutions which serve as models for disk shaped galaxies and which are used to study some physical properties of these galaxies. The so called Morgan \& Morgan disk solutions, introduced in \cite{mm69}, are important in this context. In \cite{rap12} the authors construct comparable axially symmetric solutions in Newtonian gravity with general relativistic corrections. Surprisingly these general relativistic corrections account for changes of the solutions far from the galaxy core -- a region where it was expected that Newtonian gravity describes the physics well and general relativistic effects do not play a significant role. This observation adds to the motivation of studying axially symmetric configurations of collisionless particles in the fully general relativistic picture. \par
The present article generalises \cite{akr14} to the case of charged particles, i.e.~ solutions of the EVM-system are constructed by perturbing off a non-trivial, spherically symmetric, static solution of the Vlasov-Poisson system. It is assumed that the particles are charged with a particle charge $q$, i.e.~an electro-magnetic field is included into the framework. A priori this can be done in two different ways. Either one considers $q$ as a third (a priori small) parameter which ``turns on'' charge. In this case one still perturbs off a spherically symmetric, static, uncharged solution of the Vlasov-Poisson system. The other way is to use the fact that in the non-relativistic limit the Maxwell equations reduce to the Poisson equation as well and one perturbs around a charged solution of the Vlasov-Poisson system. It turns out that the first approach is easier from a technical point of view since the operator $\mathfrak F$ that the implicit function theorem will be applied to is changed only insignificantly by the included Maxwell equations. However, the result would be restricted to small particle charge parameters $q$. In the second approach arbitrary values $0 \leq q < m_p$ of the particle charge parameter can be treated, where $m_p$ denotes the mass of the particles. In this case the operator $\mathfrak F$ has additional terms. In this article the second approach is presented. \par
In an axially symmetric, static setting the EVM-system reduces to a system of coupled, non-linear Poisson equations in different dimensions and a first order PDE. The solution of this system consists in a collection of functions which we denote $\zeta$. For the construction of a well defined solution operator $\mathfrak F$ one has to assure for that the source terms of these Poisson equations are sufficiently regular. 
However, after the variable substitution $A_\varphi = \varrho^2 a$ one obtains for the $\varphi$-component of the electro-magnetic four potential $A$ the equation
\begin{equation} \label{intro_poisson_a}
\Delta_5 a = \frac{2}{1+h} \frac{a \partial_\varrho h}{\varrho} + \frac{2}{4\pi^2 c^2} \frac{a \partial_\varrho \nu}{\varrho} + \dots.
\end{equation}
On the right hand side only some a priori problematic terms are written out explicitly. The functions $\nu$ and $h$ are part of the collection $\zeta$ of solution functions of the EVM-system. These terms are a priori problematic because they are singular at the axis $\varrho = 0$. \par
Looking a bit closer one notices that the right member of equation (\ref{intro_poisson_a}) is not singular if $h$ and $\nu$ are axially symmetric functions of a certain regularity. However, by dividing by $\varrho$ one ``looses derivatives''. For this reason the function space $\mathcal X$ has to be chosen such that the individual functions of the collection $\zeta$ have a hierarchy in regularity. For equation (\ref{intro_poisson_a}) for example one needs that $h$ and $\nu$ are of higher regularity than $a$. \par
This article is a generalisation of \cite{akr14} and the proof follows the same scheme. Including charge into the framework does not only increase the number of equations in the system but it also increases significantly the number of terms in each equation. Some of these terms require some care in the analysis but clearly not all of them. Still all required properties of the system have to be checked term by term. In order to make the presentation more concise this article resorts more to shorthands and schematic or symbolic notation than \cite{akr11, akr14}. \par
In the next section the EVM-system will be introduced. Then, in Section \ref{sect_result}, the result of this article will be stated and an outline of the proof will be given. The rest of the article is devoted to the introduction of the technical setup, the definition of the relevant objects, i.e.~function spaces and solution operators, and the proofs of important properties of these operators.

\section{The Einstein-Vlasov-Maxwell system}

A solution of the Einstein-Vlasov-Maxwell system (EVM-system) for particles with mass $m_p\geq 0$ and charge $0 \leq q < 1$ is a Lorentzian metric $g\in T^*\mathscr M \otimes T^*\mathscr M$ defined on a four dimensional manifold $\mathscr M$, a particle distribution function $f \in C^1(T\mathscr M; \mathbb R_+)$, defined on the tangent bundle of $\mathscr M$, and an electro-magnetic field tensor $F\in\Lambda^2(T \mathscr M)$ such that the EVM-system,
\begin{align}
G_{\mu\nu} &= \frac{8\pi}{c^4} \left(T_{\mu\nu} + \tau_{\mu\nu}\right), \label{eq_einstein} \\
T_{\mu\nu} &= g_{\mu\alpha} g_{\nu\beta} \frac{c}{m_p} \int_{\mathscr P_x} f(x,p) p^\alpha p^\beta \, \mathrm{dvol}_{\mathscr P_x}, \label{eq_em_tensor} \\
\tau_{\mu\nu} &= \frac{1}{4\pi} \left(-\frac 1 4 g_{\mu\nu} F_{\alpha\beta} F^{\alpha\beta} + F_{\nu\alpha} F_{\mu}^{\;\;\alpha} \right), \label{el_em_tensor} \\
\mathfrak T (f) &=0, \label{eq_vlasov} \\
\mathrm dF &= 0, \label{maxwell_eq_1} \\
\nabla_\alpha F^{\alpha\beta} &= - 4\pi qJ^\beta,\quad J^\beta =\frac 1 c \int_{\mathscr P_x} f(x,p)p^\beta \mathrm{dvol}_{\mathscr P_x}, \label{maxwell_eq_2} 
\end{align}
is satisfied. Here $G_{\mu\nu}$ is the Einstein tensor and we choose units such that $G=1$ ($G$ is the gravitational constant) but we leave $c$ as parameter in the system. \par
We give a brief explanation of the involved quantities, consult however e.g.~\cite{sz14} for a more detailed introduction to the EVM-system. The particle distribution function $f=f(x,p)$ describes the particle number density at a certain point in $x \in \mathscr M$ with a certain four-momentum $p \in T_x\mathscr M$. The particle number can be obtained via integration. The quantity $m_p$, defined by the relation
\begin{equation}  \label{m_s_r}
g_{\mu\nu}(x) p^\mu p^\nu = -c^2 m_p^2, \qquad x \in \mathscr M, p \in T_x\mathscr M
\end{equation}
is interpreted as the particles' rest mass. It can be shown that it stays constant along the characteristic curves of the Vlasov equation (\ref{eq_vlasov}). Consequently the particle distribution function $f$ describing an ensemble of particles where all particles have the same rest mass $m_p$ can be assumed to be supported on the mass shell $\mathscr P_{m_p}$, a seven dimensional submanifold of $T\mathscr M$ which is defined to be
\begin{equation}
\mathscr P_{m_p} = \{(x,p)\in T\mathscr M \,:\, g_{\mu\nu}(x)p^\mu p^\nu = -c^2 m_p^2, \; p\,\mathrm{is\,future\,pointing}\}.
\end{equation}
In the remainder of this article we assume $m_p=1$ for all particles, and we denote the corresponding mass shell simply by $\mathscr P$. The volume form $\mathrm{dvol}_{\mathscr P_x}$ on the mass shell fibre $\mathscr P_x$ over $x\in \mathscr M$ is given by
\begin{equation}
\mathrm{dvol}_{\mathscr P_x} = \frac{\sqrt{|\det(g_{\mu\nu}(x))|}}{-p_0}\, \mathrm d p^1 \wedge \mathrm dp^2 \wedge \mathrm dp^3,
\end{equation}
and the transport operator $\mathfrak T$ is given by
\begin{equation} \label{def_transport_op}
\mathfrak T = p^\mu \partial_\mu + \left(q F^\gamma{}_{\mu} \, p^\mu - \Gamma^\gamma_{\alpha\beta} p^\alpha p^\beta\right) \partial_{p^\gamma}.
\end{equation}
It is tangent to any mass shell $\mathscr P$ \cite{sz14}. \par
Assume that we have a solution $(g,f,F)$ of the EVM system and that on $\mathscr M$ we have coordinates $t$, $x^1$, $x^2$, $x^3$, where $t$ is the time coordinate. Assume further that $\partial_t$ is a Killing field. Then the solution is asymptotically flat if the boundary conditions
\begin{equation} \label{bc_as_flat}
\lim_{|x|\to \infty} g = \eta, \quad \lim_{|x|\to \infty} f = 0, \quad \lim_{|x|\to \infty} F = 0
\end{equation}
are satisfied, where $\eta$ denotes the Minkowski metric.

\section{The result} \label{sect_result}

In this article we prove the following result.

\begin{theorem} \label{main_theorem}
There exist asymptotically flat, stationary solutions $(g, f, F) \in (T^*\mathscr M \otimes T^*\mathscr M) \times C_c^1 (\mathscr P; \mathbb R_+) \times \Lambda^2(\mathscr M)$ of the EVM-system (\ref{eq_einstein})--(\ref{maxwell_eq_2}) with particle charge parameters $q\in [0,1)$, which are axially symmetric but not spherically symmetric. Such a solution has no toroidal magnetic field and it has a poloidal magnetic field if and only if the solution is not static, i.e.~rotating.
\end{theorem}

\begin{proof}
The proof which is given at this place is rather an outline of the poof, the technical details are given in the subsequent sections. The proof follows the same structure as in \cite{akr14} where the existence of stationary, rotating, axially symmetric solutions is proved for uncharged particles. Each step is however a bit more involved and some arguments have to be formulated differently due to the additional Maxwell equations. We comment on the modifications in the respective sections. \par
{\it Step 1: Elimination of the Vlasov equation.} For the particle distribution function we use the ansatz $f(x,p) = \phi(E(x,p))\psi(\lambda, L(x,p))$, see (\ref{ansatz_f}) below. So the particle distribution depends only on the particle energy $E(x,p)$ and the $z$-component of the angular momentum $L(x,p)$, see the definitions (\ref{formula_e}) and (\ref{formula_l}) below. Since the quantities $E$ and $L$ are conserved along its characteristics the Vlasov equation is automatically satisfied for such an ansatz, cf.~Section \ref{sect_char} below. Furthermore, we introduce a parameter $\lambda$ which ``turns on'' the dependency of $f$ on $L$. This means that if $\lambda = 0$ then $\psi \equiv 1$, i.e.~for each value of the $z$-component of the angular momentum there are equally many particles. \par
{\it Step 2: Reduction of the remaining system.} First we express the EVM-system (\ref{eq_einstein})--(\ref{maxwell_eq_2}) in cylindrical coordinates. The assumptions that the solution is asymptotically flat, axially symmetric, and time independent yield simplifications of the system of equations. We call this simplified system the {\em reduced EVM-system}, cf.~Definition \ref{def_red_evm} below, and it is stated in Section \ref{sect_red_sys}, equations (\ref{final_eq_nu})--(\ref{bc_center}), below, where any value of $c\in (0,\infty)$ is admitted. The solution of the reduced EVM-system is determined by the collection $\zeta = (\nu, h, \xi, \omega, A_t, a) \in \mathcal X$ of six functions, defined in a suitably chosen function space $\mathcal X$ (defined in Section \ref{sect_function_space} below). Proposition \ref{prop_equivalent} below states that a solution of the reduced EVM-system with any parameter $c$ can be converted into an axially symmetric, stationary solution of the EVM-system with the parameter $c=1$. \par
{\it Step 3: Introduction of the solution operator $\mathfrak F$.} A solution of the reduced EVM-system with parameters $\gamma:= c^{-2}, \lambda \in  [0,1) \times (-1,1)$ is then obtained as perturbation of a spherically symmetric solution of the Vlasov-Poisson system. This spherically symmetric solution of the Vlasov-Poisson system we denote by $\zeta_0\in\mathcal X$. \par
To this end in Section \ref{sect_def_f} an operator $\mathfrak F: \mathcal X \times [0,1) \times (-1,1) \to \mathcal X$ with the following properties is defined. Firstly, a collection of functions $\zeta\in\mathcal X$ is a solution of the reduced EVM-system with parameters $\gamma$, $\lambda$ if and only if $\mathfrak F[\zeta; \gamma, \lambda] = 0$. (The ``if''-direction is essential.) Secondly, $\mathfrak F[\zeta_0; 0, 0] = 0$. In Section \ref{sect_well_defined} we show that this operator is well defined. The mentioned properties are shown in Proposition \ref{prop_consistent} and Lemma \ref{lem_n_zero} below. \par
{\it Step 4: Application of the implicit function theorem.} The aim is to apply the implicit function theorem on Banach spaces, cf.~for example \cite[Theorem 15.1]{d85}. This theorem implies the existence of $\delta>0$ such that there exists a mapping $\mathfrak Z : [0,\delta) \times (-\delta, \delta) \to \mathcal X$ such that for all $(\gamma,\lambda)\in [0,\delta)\times (-\delta,\delta)$ we have 
\begin{equation}
\mathfrak F(\mathfrak Z(\gamma,\lambda); \gamma, \lambda) = 0,
\end{equation}
i.e.~$\mathfrak Z(\gamma,\lambda)$ is a solution of the reduced EVM-system with parameters $\gamma, \lambda$. This solution $\mathfrak Z(\gamma,\lambda)$ then gives rise to a solution of the EVM-system with the asserted properties, by Proposition \ref{prop_equivalent}. \par
The implicit function theorem can be applied in this way if the operator $\mathfrak F$ is continuous at $(\zeta_0; 0, 0)$, if its Fr\'echet derivative $\mathfrak L := D\mathfrak F[\zeta_0;0,0]: \mathcal X \to \mathcal X$ at the point $(\zeta_0; 0, 0)\in \mathcal X \times [0,\delta) \times (-\delta,\delta)$ exists and is continuous, and if this Fr\'echet derivative $\mathfrak L$ is a bijection. These properties are established in Section \ref{sect_frechet_dir}. Proposition \ref{prop_appli} below contains the details how it is made sure that the boundary conditions for an asymptotically flat solutions are satisfied. \par
{\it Step 5: Characterisation of the electro-magnetic field.} The assertion that the solution comprises a poloidal magnetic field if and only if the solution is rotating follows from the structure of the reduced EVM-system, see Remark \ref{rem_mag_rot}. For the assertion that there is no toroidal magnetic field, see Lemma \ref{lem_no_tor}. \par
\end{proof}

\section{Axial symmetry} \label{sect_axial}

Let $x^i$, $i=1,\dots,n$ be coordinates on $\mathbb R^n$. A function $f:\mathbb R^n \to \mathbb R$ is axially symmetric around the $x^n$-axis if and only if there exists a function $\hat f: [0,\infty) \times \mathbb R \to \mathbb R$ such that 
\begin{equation}
f\left(x^1,\dots,x^n\right) = \hat f\left(\varrho(x^1, \dots, x^{n-1}), x^n\right), 
\end{equation}
where
\begin{equation} \label{def_rho}
\varrho(x^1, \dots, x^{n-1}) := \sqrt{\left(x^1\right)^2 + \dots + \left(x^{n-1}\right)^2}.
\end{equation}
By abuse of notation, we will use the same symbol for the original function on $\mathbb R^2$, $\hat f$ in this example, and the induced axially symmetric functions $f$ on $\mathbb R^n$ for different dimensions $n$.

\begin{rem}\label{rem_even}
At some places in the analysis presented in this article it will be useful to view an axially symmetric function $f:\mathbb R^n \to \mathbb R$ as a function in $\varrho$ and $z$ defined on $\mathbb R^2$, by extending it as even function to negative values of $\varrho$. The obtained function on $\mathbb R^2$ then has the same regularity as the axially symmetric function on $\mathbb R^n$.
\end{rem}

We now introduce a coordinate gauge and the functions in terms of which we will formulate the reduced EVM-system. Consider the four dimensional manifold $\mathscr M$ which is assumed to be homeomorphic to $\mathbb R^4$ and which is equipped with the cylindrical coordinates $t$, $\varrho$, $z$, $\varphi$. A stationary Lorentzian metric is characterised by the four time independent, axially symmetric functions $\nu,\mu,\omega:\mathscr M \to \mathbb R$ and $H: \mathscr M \to \mathbb R_+$. It can be written in the form
\begin{equation} \label{ansatz_metric}
g = -c^2 e^{\frac{2\nu(\varrho, z)}{c^2}}\mathrm dt^2 + e^{2\mu(\varrho, z)}\mathrm d\varrho^2 + e^{2\mu(\varrho, z)}\mathrm dz^2 + \varrho^2H(\varrho, z)^2 e^{-\frac{2\nu(\varrho, z)}{c^2}} \left(\mathrm d\varphi - \omega(\varrho, z) \mathrm dt\right)^2,
\end{equation}
cf.~\cite{b72} for details. \par
The electro-magnetic field tensor $F$ is given as the exterior derivative of the electro-magnetic four potential $A\in\Lambda^1(\mathscr M)$, i.e.~$F = dA$. With respect to the coordinate co-basis of $t$, $\varrho$, $z$, $\varphi$ the electro-magnetic potential $A$ takes the form 
\begin{equation}
A = A_t \mathrm dt + A_\varrho \mathrm d\varrho + A_\varphi \mathrm d\varphi + A_z \mathrm dz.
\end{equation}
We assume that all components are time independent and axially symmetric. \par
In terms of the electro-magnetic field tensor $F$ the electric field $E\in \Lambda^1(\mathscr M)$ and the magnetic field $\mathscr B \in \Lambda^1(\mathscr M)$ are defined as follows. The electric field $E$ is defined by the splitting $F = E \wedge \mathrm dt + B$, where the two form $B$ includes no term with $\mathrm dt$. The magnetic field is defined by the splitting $\star F = \mathscr E - \mathscr B \wedge \mathrm dt$, where $\star : \Lambda^2(\mathscr M) \to \Lambda^2(\mathscr M)$ is the Hodge star operator and $\mathscr E$ is a two-form with no $\mathrm dt$-term. Cf.~\cite{f12} for details. Define $\beta := \partial_z A_{\varrho} - \partial_\varrho A_{z}$. Then a calculation yields that the toroidal magnetic field component $\mathscr B_\varphi$ takes the form
\begin{equation} \label{tor_mag}
\mathscr B_\varphi = 2c e^{-2\mu} \varrho H \beta,
\end{equation}
and the poloidal magnetic field components, $\mathscr B_\varrho$ and $\mathscr B_z$, contain only the $t$- and the $\varphi$-component of $A$. In fact a calculation yields
\begin{align}
\mathscr B_\varrho &=- \frac{2 e^{- 2\nu/c^2}}{c \varrho H} \left( c^2 e^{4\nu/c^2} A_{\varphi,z} - \varrho^2 H^2 \omega (A_{t,z} + \omega A_{\varphi,z}) \right), \\
\mathscr B_z &= \frac{2 e^{- 2\nu/c^2}}{c \varrho H} \left( c^2 e^{4\nu/c^2} A_{\varphi,\varrho} - \varrho^2 H^2 \omega (A_{t,\varrho} + \omega A_{\varphi,\varrho}) \right).
\end{align}
\par
Next we introduce the parameter $\gamma = \frac{1}{c^2}$ and the orthonormal frame $e_a = e_a{}^\alpha \partial_\alpha$, $\alpha = t,\varrho, z, \varphi$, where the non-trivial matrix elements are
\begin{equation} \label{frame_matrix}
e_0{}^t = e^{-\gamma \nu}, \quad e_0{}^\varphi = e^{-\gamma \nu} \omega, \quad e_1{}^\varrho = e^{-\mu}, \quad e_2{}^z = e^{-\mu}, \quad e_3{}^\varphi = \frac{e^{\gamma \nu}}{\varrho H}.
\end{equation}
The corresponding co-frame reads $\alpha^a = e^a{}_\alpha \mathrm dx^\alpha$, where $(e^a{}_\alpha) = (e_a{}^\alpha)^{-1}$ (the inverse matrix), 
and via the relation $p^\mu \partial_\mu = v^\mu e_\mu$ this frame introduces the new momentum variables $v^0, v^1, v^2, v^3$, given by
\begin{equation} \label{def_frame}
v^0 = e^{\gamma\nu} p^t, \quad v^1 = e^\mu p^\varrho, \quad v^2 = e^\mu p^z, \quad v^3 = \varrho H e^{-\gamma\nu} \left(p^\varphi - \omega p^t\right).
\end{equation} 
In the remainder of this article we work with the coordinates 
\begin{equation} \label{the_coords}
t \in \mathbb R, \quad \varrho \in [0, \infty), \quad \varphi \in [0,2\pi), \quad z \in \mathbb R, \quad (v^0, v^1, v^2, v^3) \in \mathbb R^4
\end{equation}
on the tangent bundle $T\mathscr M$. In these frame coordinates the mass shell relation (\ref{m_s_r}) becomes
\begin{equation}
-c^2 = -c^2 \left(v^0\right)^2 + \left(v^1\right)^2 + \left(v^2\right)^2 + \left(v^3\right)^2
\end{equation}
and on $\mathscr P$ we consequently have
\begin{equation} \label{mass_shell_frame}
v^0 = \sqrt{1+\gamma |v|^2}, \quad \mathrm{where}\, |v| = \sqrt{\left(v^1\right)^2 + \left(v^2\right)^2 + \left(v^3\right)^2}.
\end{equation}

\section{The method of characteristics} \label{sect_char}

The Vlasov equation (\ref{eq_vlasov}) can be dealt with by the method of characteristics which is now described. 

\begin{lemma} \label{lem_conserved_quan}
The quantities $E$ and $L$, defined on the tangent bundle $T \mathscr M$, by
\begin{align}
L &:= \varrho H e^{ -\gamma\nu} v^3 - q A_\varphi, \label{formula_l} \\
E &:= \frac{e^{\gamma\nu} v^0 - 1}{\gamma} + \omega \varrho H e^{ -\gamma\nu} v^3 + q A_t, \label{formula_e}
\end{align}
are conserved along the characteristic curves of the Vlasov equation, i.e.
\begin{equation}
\mathfrak TE = 0, \quad \mathfrak TL = 0.
\end{equation}
\end{lemma}

\begin{proof}
The assertion of this lemma can be shown via a direct calculation and it is moved to the appendix.
\end{proof}

\begin{rem}
Unlike the uncharged case, in the charged case the characteristic curves of the Vlasov equations are not the lifts of the geodesics to $T\mathscr M$. Consequently the conserved quantities cannot be obtained by $g(X, p)$, where $X$ is a Killing vector field and $p$ is the canonical momentum. However, this structure can still be recognised in the present case. If we define 
\begin{align}
\tilde E &:= -g(\partial_t, p), \label{tilde_e} \\
\tilde L &:= g(\partial_\varphi, p), \label{tilde_l}
\end{align}
it turns out that the quantities $E$ and $L$ can be obtained from $\tilde E$ and $\tilde L$ by taking into account a suitable correction due to the electro-magnetic field. We have
\begin{equation}
E = \tilde E - \frac{1}{\gamma} + qA_t, \qquad L = \tilde L - qA_\varphi.
\end{equation}
\end{rem}

\begin{cor} \label{cor_char}
Every function $f:\mathscr P \to \mathbb R_+$ which can be expressed as
\begin{equation} \label{product_strucutre}
f(t,\varrho,\varphi,z,v^0,v^1,v^2,v^3) = \phi(E) \tilde \psi(L)
\end{equation}
with some functions $\phi, \tilde \psi \in C^1(\mathbb R ; \mathbb R_+)$, solves the Vlasov equation (\ref{eq_vlasov}) and is axially symmetric and time independent.
\end{cor}

\begin{proof}
Since $\mathfrak T E = \mathfrak T L = 0$ we have by the chain rule $\mathfrak Tf = 0$. The remaining asserted properties of $f$ are inherited from the metric functions $\nu$, $\mu$, $H$, and $\omega$.
\end{proof}

A more general statement than Corollary \ref{cor_char} is true, for ansatz functions that do not have the product structure (\ref{product_strucutre}). The corollary is however stated this way because in this article only ansatz functions of the form (\ref{product_strucutre}) are considered. \par
From now on we work with the ansatz
\begin{equation} \label{ansatz_f}
f(x,v) = \phi\left(E \right)\psi(\lambda, L),
\end{equation}
where $E$ and $L$ are the conserved quantities, given in (\ref{formula_e}) and (\ref{formula_l}), respectively, and $\lambda \in [0,1]$ is the parameter which ``turns on'' anisotropy in momentum of the particle distribution. The functions $\phi$ and $\psi$ are assumed to fulfil the assumptions listed below. For an integrable function $U$ and $\phi \in C^1(\mathbb R;\,\mathbb R_+)$, where $\mathrm{supp}(\phi)\subset (-\infty, E_0]$ for some $0\leq E_0 < \infty$, we define
\begin{align}
\rho_U(r) := \int_{\mathbb R_v^3} \phi\left(\frac{|v|^2}{2} + U(r)\right) \,\mathrm dv^1 \mathrm dv^2 \mathrm dv^3,  \label{def_rho_u} \\
\alpha_U(r) :=  \int_{\mathbb R_v^3} \phi'\left(\frac{|v|^2}{2} + U(r)\right) \,\mathrm dv^1 \mathrm dv^2 \mathrm dv^3. \label{def_alpha_u}
\end{align} 
We assume that the functions $\phi$ and $\psi$ in (\ref{ansatz_f}) have the following properties.
\begin{enumerate}
\item $\phi\in C^2(\mathbb R)$ and there exists $E_0 > 0$ such that $\phi(E)=0$ for $E\geq  E_0$ and $\phi(E) > 0$ for $E < E_0$. \label{cond_phi}
\item The ansatz $f(x,v)=\phi\left(\frac 12 |v|^2 + U_N(x)\right), x, v \in\mathbb R^3$, leads to a compactly supported, spherically symmetric steady state $(f_N, U_N)$ of the Vlasov-Poisson system for particles with mass $1-q^2$, i.e., there exists a solution $U_N \in C^2(\mathbb R^3)$, of the equation $\Delta U_N = 4\pi (1-q^2) \rho_{N}(x)$, $U_N(0)=0$, where we used the shorthand $\rho_N := \rho_{U_N}$. This solution is spherically symmetric, $U_N(x) = U_N(|x|)$, and the support of $\rho_N \in C_c^2(\mathbb R^3)$ is the closed ball $\overline B_{R_N}(0)$ where $U_N(R_N)=E_0$ and $U_N(r) < E_0$ for $0\leq r < R_N < \infty$, and $U_N(r) > E_0$ for $r> R_N$.
\item We have $6+ 4\pi (1-q^2) r^2 \alpha_N(r) > 0$ for all $r\in[0,\infty)$.
\item $\psi\in C_c^\infty (\mathbb R^2)$ is compactly supported, $\psi \geq 0$, $\partial_L \psi(\lambda, 0) = 0$ for $\lambda \in \mathbb R$, and $\psi(0, L) = 1$ on an open neighbourhood of the set $\{L = L_N(x,v) \, | \, (x,v) \in \mathrm{supp}(f_N)\}$, where $L_N := \varrho v^3$ is the $z$-component of the Newtonian angular momentum.
\end{enumerate}

\begin{lemma}
There exist ansatz functions $\phi \in C_c^2(\mathbb R)$ and $\psi\in C_c^\infty((-1/2, 1/2) \times \mathbb R)$ satisfying the upper conditions.
\end{lemma}

\begin{proof}
Consider the polytropes $\phi(E) = [E_0-E]_+^k$ for $k\in[2,7/2)$. Condition (1) is clearly satisfied. Condition (2) is also satisfied, cf.~\cite{bfh86, rr00}.\par
By the same proof as for \cite[Lemma 7.1]{akr11} it can be shown that the third condition is satisfied for polytropes with exponent $k$ sufficiently close to $7/2$. To this end one uses the equation $\Delta U_N = 4\pi (1-q^2) \rho_N$ instead of $\Delta U_N = 4\pi \rho_N$. Then merely the constant $4\pi$ has to be replaced by $4\pi(1-q^2)$ in the proof of \cite[Lemma 7.1]{akr11}. It is essential that $1-q^2 > 0$, the precise value is however irrelevant for the argument.
\end{proof}

\section{The reduced system of equations} \label{sect_red_sys}

Before the reduced system of equations is presented some notation and shorthands shall be introduced. Partial derivatives $\partial_\varrho \nu$, $\partial_z A_\varphi$, etc.~will be denoted as $\nu_{,\varrho}$, $A_{\varphi,z}$, etc. We define the functions $\xi, h, a$ by the following changes of variables:
\begin{align}
\xi &= \mu + \gamma \nu, \\
H &= 1 + h, \\
A_\varphi &= \varrho^2 a.
\end{align}
Further, we call $(\nu, h, \xi, \omega, A_t, a)$ the {\em solution functions} and in the remainder of this article we will use the shorthand
\begin{equation}
\zeta := (\nu, h, \xi, \omega, A_t, a).
\end{equation}
We do not include the components $A_\varrho$ and $A_z$ of the four-potential $A$ into the solution functions $\zeta$ since it will turn out that in the current setting they must vanish everywhere, cf.~Lemma \ref{lem_no_tor} below. \par
In \cite{akr11, akr14}, where the existence of axially symmetric solutions of the Einstein-Vlasov system with uncharged particles is proven, a reduced system of equations is considered as well. The reduced EVM-system presented below coincides with the reduced system in \cite{akr14} if the charge parameter $q$ is set to zero. When the Maxwell equations are added to the framework not only the number of equations increases but also the number of terms in the Einstein equations increases by a multiple. For this reason, below, we are going to introduce {\em source functions} to collect these terms. This allows to present the reduced system in a compact way and also facilitates the presentation of the subsequent analysis. Moreover, we will introduce {\em matter functions} which basically consist in combinations of components $T_{\mu\nu}$ of the Vlasov part of the energy momentum tensor, as in \cite{akr14}. \par
In the subsequent analysis it will be necessary to show different properties of the matter functions and the source functions, like regularity with respect to the coordinates $\varrho$ and $z$, decay properties, symmetries, or Fr\'echet differentiability with respect to the solution functions $\zeta$. This means that at some occasions the source functions and the matter functions have to be seen as functions of $\varrho$ and $z$ which are parameterised by the solution functions. At other occasions they have to be seen as functions which take both the coordinate $\varrho$ and the solution functions $\zeta$ (and their derivatives) as arguments. Moreover, for the analysis of the matter functions several different integral representations will be necessary. In order to give a clear presentation we deem it favourable to resort to symbolic notation in a larger extent than in \cite{akr14}. \par
Now we define the matter functions. These matter functions depend on the solution functions $\nu$, $h$, $\xi$, $\omega$, $A_t$, $a$. At different places in this article we want to see them either as functions taking the evaluated solution functions as argument ($M_i^{(\gamma,\lambda)}$ below) or as families of functions which are parameterised by the solution functions ($\mathfrak M_i[\zeta; \gamma, \lambda]$ below) and which only depend on $(\varrho,z)$. We define
\begin{align}
M_1^{(\gamma,\lambda)}(\varrho, \zeta) &:= 4\pi e^{2(\xi-\gamma\nu)} \int_{\mathbb R_v^3} \phi(E)\psi(\lambda, L) \frac{1 + 2 \gamma |v|^2}{\sqrt{1+ \gamma |v|^2}} \, \mathrm d^3v, \label{def_hat_m_1}  \\
M_2^{(\gamma,\lambda)}(\varrho, \zeta) &:= 8\pi \gamma^2 (1+h) e^{2(\xi-\gamma\nu)} \int_{\mathbb R_v^3} \phi(E)\psi(\lambda, L) \frac{(v^1)^2+(v^2)^2}{\sqrt{1+\gamma |v|^2}} \, \mathrm d^3v, \label{def_hat_m_2}\\
M_4^{(\gamma,\lambda)}(\varrho, \zeta) &:= -\frac{16\pi \gamma}{\varrho(1+h)} e^{2\xi - 4\gamma\nu}\int_{\mathbb R_v^3} \phi(E)\psi(\lambda, L) v^3\,\mathrm d^3v, \label{def_hat_m_4}  \\
M_5^{(\gamma,\lambda)}(\varrho, \zeta) &:= 4\pi q e^{2\xi - 3\gamma\nu} \int_{\mathbb R_v^3} \phi(E)\psi(\lambda, L) \left(e^{2\gamma\nu} + \frac{\gamma \varrho (1+h) \omega v^3}{\sqrt{1+\gamma |v|^2}}\right) \, \mathrm d^3v, \label{def_hat_m_5} \\
M_6^{(\gamma,\lambda)}(\varrho, \zeta) &:= -\frac{4\pi q \gamma (1+h)}{\varrho} e^{2\xi-3\gamma\nu} \int_{\mathbb R_v^3} \phi(E)\psi(\lambda, L) \frac{v^3}{\sqrt{1+\gamma |v|^2}}\, \mathrm d^3v, \label{def_hat_m_6}
\end{align}
where $\mathrm d^3v = \mathrm d v^1 \mathrm dv^2 \mathrm dv^3$ and $E$ and $L$ are seen as functions of $\varrho$, $\zeta$ and $v^1,v^2, v^3$, according to the formulas (\ref{formula_e}) and (\ref{formula_l}) whereas $\xi$, $\nu$, $h$, $\varrho$ are seen as variables. Moreover let
\begin{equation} \label{def_all_ms}
\mathfrak M_i[\zeta; \gamma, \lambda](\varrho, z) := M_i^{(\gamma,\lambda)}(\varrho, \zeta(\varrho, z)), \quad i = 1, 2, 4, 5, 6. 
\end{equation}
We remark that if $\psi$ is even in $L$, then 
\begin{equation} \label{m4_m6_0}
M_4^{(\gamma,\lambda)}(\varrho,\zeta) = M_6^{(\gamma,\lambda)}(\varrho,\zeta) = 0, \quad \mathrm{if}\;\omega=0.
\end{equation}
This follows immediately since the integrand in $M_4^{(\gamma,\lambda)}$ and $M_6^{(\gamma,\lambda)}$ is antisymmetric in $v^3$. \par
In the same spirit as the matter functions we define for $\gamma\in [0,1]$ the source functions $g_i^{(\gamma)} : \mathbb R^{19} \to \mathbb R$, $i=1, \dots, 6$. The source functions take the solution functions $\zeta_i$, $i=1,\dots,6$ and their derivatives $\zeta_{i,\varrho}$ and $\zeta_{i,z}$, $i=1,\dots,6$ as separate arguments, i.e.~they are considered as independent variables. We denote
\begin{align*}
\zeta_{,\varrho} = (\zeta_{1,\varrho}, \dots, \zeta_{6,\varrho}) = (\partial_{\varrho} \zeta_1, \dots, \partial_\varrho \zeta_6), \quad 
\zeta_{,z} = (\zeta_{1,z}, \dots, \zeta_{6,z}) = (\partial_{z} \zeta_1, \dots, \partial_z \zeta_6).
\end{align*}
Then the source functions are defined to be
\begin{align}
g_1^{(\gamma)}(\varrho, \zeta, \zeta_{,\varrho}, \zeta_{,z}) &:= -\frac{ h_{,\varrho}  \nu_{,\varrho} +  h_{,z}  \nu_{,z}}{1+ h} + \frac{\varrho^2}{2}(1+ h)^2 e^{-4\gamma \nu} \left( \omega_{,\varrho}^2 +  \omega_{,z}^2\right) \label{for_g_1} \\
&\quad -\gamma^2 e^{-2\gamma\nu}  \left( ( A_{t,\varrho} + 2\omega\varrho  a + \omega \varrho^2  a_{,\varrho})^2 + ( A_{t,z} + \omega \varrho^2  a_{,z})^2 \right) \nonumber \\
&\quad - \gamma \frac{e^{2\gamma \nu}}{(1+ h)^2} \left((2 a+\varrho  a_{,\varrho})^2 + \varrho^2  a_{,z}^2 \right), \nonumber
\end{align}
\begin{align}
&g_3^{(\gamma)}(\varrho, \zeta, \zeta_{,\varrho}, \zeta_{,z}) := \left(\left(1 + \partial_{\varrho}(\varrho  h)\right)^2 + \varrho^2  h_{,z}^2\right)^{-1} \label{for_g_3} \\
&\times \bigg( (1+\partial_{\varrho} (\varrho  h)) \nonumber \\
&\qquad \times \left[ \frac{\varrho}{2}( h_{\varrho\varrho} -  h_{zz}) +  h_{,\varrho} -\gamma^2 (1+ h) \varrho ( \nu_{,z}^2 -  \nu_{,\varrho}^2) - \gamma\varrho^3 (1+ h)^3 e^{-4\gamma \nu} ( \omega_{,\varrho}^2 -  \omega_{,z}^2) \right] \nonumber \\
&\qquad + \varrho  h_{,z} \left[\partial_{,\varrho} (\varrho  h_{,z}) + 2 \gamma^2 (1+ h) \varrho  \nu_{,\varrho}  \nu_{,z} + \frac 12 \gamma e^{-4 \gamma  \nu} \varrho^3 (1+ h)^3  \omega_{,\varrho}  \omega_{,z}\right] \nonumber  \\
&\qquad - 2\gamma^3 e^{-2 \gamma  \nu} (1+ h) \varrho^2  h_{,z} \left(\left( A_{t,\varrho} + 2\varrho  \omega  a + \varrho^2  \omega  a_{,\varrho} \right) \left( A_{t,z} + \varrho^2  \omega  a_{,z}\right) \right) \nonumber \\
&\qquad + \gamma^3 e^{-2 \nu \gamma} (1+ h) \varrho (1+\partial_{,\varrho}(\varrho  h)) \left(\left( A_{t,z}+ \varrho^2  \omega  a_{,z} \right)^2 - \left( A_{t,\varrho} + 2\varrho  \omega  a + \varrho^2  \omega  a_{,\varrho}\right)^2 \right) \nonumber \\
&\qquad + \gamma^2 \varrho^3 e^{2 \gamma \nu} \left(2\frac{ h_{,z}}{1 +  h} (2 a  a_{,z} + \varrho  a_{,\varrho}  a_{,z}) + \left(1 + \varrho \frac{h_{,\varrho}}{1+h} \right) \left(a_{,\varrho}^2 - a_{,z}^2\right)\right) \bigg), \nonumber
\end{align}
\begin{align}
g_4^{(\gamma)}(\varrho, \zeta, \zeta_{,\varrho}, \zeta_{,z}) &:= - \left(3\frac{ h_{,\varrho} \omega_{,\varrho} +  h_{,z}  \omega_{,z}}{1+ h} - 4\gamma (\nu_{,\varrho}  \omega_{,\varrho} +  \nu_{,z}  \omega_{,z}) \right) \label{for_g_4} \\
&\quad \;\, + 4 \gamma^2 \frac{e^{2\gamma\nu}}{(1+ h)^2} \Big( \frac{2}{\varrho}  A_{t,\varrho}  a +  A_{t,\varrho}  a_{,\varrho} +  A_{t,z}  a_{,z} + 4 \omega  a^2 \nonumber \\
&\hspace{5cm} + 2\omega  \varrho  a  a_{,\varrho} +  \omega \varrho^2  a_{,\varrho}^2 + \omega \varrho^2  a_{,z}^2\Big), \nonumber 
\end{align}
\begin{align}
&g_5^{(\gamma)}(\varrho, \zeta, \zeta_{,\varrho}, \zeta_{,z}) \label{for_g_5} \\
&\;\; := 2\gamma \left(\nu_{,\varrho} A_{t,\varrho} + \nu_{,z} A_{t,z}\right) +4\gamma\omega \left(2\varrho  \nu_{,\varrho}  a + \varrho^2 \nu_{,\varrho}  A_{,\varrho} + \varrho^2 \nu_{,z}  a_{,z}\right) \nonumber \\
&\qquad - \frac{h_{,\varrho} A_{t,\varrho} + h_{,z} A_{t,z}}{1+ h} -2\omega \frac{2\varrho  h_{,\varrho}  a + \varrho^2 h_{,\varrho} a_{,\varrho} + \varrho^2  h_{,z}  a_{,z}}{1+ h} \nonumber \\
&\qquad - \left(2\varrho  a \omega_{,\varrho} + \varrho^2 a_{,\varrho} \omega_{,\varrho} + \varrho^2 a_{,z} \omega_{,z}\right) -2 (2\omega  a + \varrho  \omega  a_{,\varrho}) \nonumber \\
&\qquad -\gamma \omega \varrho^2 (1+ h)^2 e^{-4\gamma\nu} \left( \omega_{,\varrho}( A_{t,\varrho} + 2 \varrho  \omega  a + \varrho^2  \omega  a_{,\varrho}) +  \omega_{,z}( A_{t,z} + \varrho^2  \omega  a_{,z})\right), \nonumber
\end{align}
\begin{align}
g_6^{(\gamma)}(\varrho, \zeta, \zeta_{,\varrho}, \zeta_{,z}) &:= \gamma (1+ h)^2 e^{-4\gamma \nu} \left( \omega_{,\varrho} ( A_{t,\varrho} + 2\varrho  \omega  a + \varrho^2  \omega  a_{,\varrho}) + \omega_{,z} ( A_{t,z} + \varrho^2 \omega  a_{,z})\right) \label{for_g_6}\\
&\qquad + \frac{\frac 2 \varrho  h_{,\varrho}  a +  h_{,\varrho}  a_{,\varrho} +  h_{,z}  a_{,z}}{1+ h} + \frac{\gamma}{4\pi^2} \left(\frac{2}{\varrho} \nu_{,\varrho}  a + \nu_{,\varrho}  a_{,\varrho} +  \nu_{,z}  a_{,z}\right). \nonumber
\end{align}
Furthermore we define
\begin{equation} \label{def_mathfrak_g}
\mathfrak g_i[\zeta; \gamma](\varrho,z) := g_i^{(\gamma)}(\varrho, \zeta(\varrho,z), \zeta_{,\varrho}(\varrho,z), \zeta_{,z}(\varrho,z)), \quad i = 1,3,\dots, 6
\end{equation}
as families of source functions which depend only on $\varrho$ and $z$ but which are parameterised by the solution functions $\zeta$. Moreover we define the operators
\begin{equation}
\Delta_n := \partial_{\varrho\varrho} + \frac{n-2}{\varrho} \partial_\varrho + \partial_{zz}, \quad n=3,4,5.
\end{equation}
As the notation indicates, these operators correspond to the Laplace operator for axially symmetric functions in three, four, and five dimensions. We consider the following boundary value problem, consisting in the Einstein equations,
\begin{align}
\Delta_3 \nu(\varrho,z) &= \mathfrak g_1[\zeta; \gamma](\varrho,z) + \mathfrak M_1[\zeta; \gamma, \lambda](\varrho,z) \label{final_eq_nu}  \\
\Delta_4 h(\varrho,z)  &= \mathfrak M_2[\zeta; \gamma, \lambda](\varrho,z), \label{eq_b} \\
\xi_{,\varrho}(\varrho,z) &= \mathfrak g_3[\zeta; \gamma](\varrho,z), \label{eq_xi} \\ 
\Delta_5 \omega(\varrho,z)  &= \mathfrak g_4[\zeta; \gamma](\varrho,z) + \mathfrak M_4[\zeta; \gamma, \lambda](\varrho,z), \label{eq_omega}
\end{align}
poloidal Maxwell equations,
\begin{align}
\Delta_3 A_t(\varrho,z) &= \mathfrak g_5[\zeta; \gamma](\varrho,z) + \mathfrak M_5[\zeta; \gamma, \lambda](\varrho,z), \label{eq_a0}\\
\Delta_5 a(\varrho, z) &= \mathfrak g_6[\zeta; \gamma](\varrho,z) + \mathfrak M_6[\zeta; \gamma, \lambda](\varrho,z), \label{final_eq_a}
\end{align}
toroidal Maxwell equations,
\begin{align}
\left(\frac{h_{,z}}{1+h} + 2(\gamma \nu_{,z}-\xi_{,z})\right)\left(A_{z,\varrho} - A_{\varrho,z}\right) + \partial_z \left(A_{z,\varrho} - A_{\varrho,z}\right) &= 0, \label{tor_maxwell_1} \\
\left(\frac 1 \varrho + \frac{h_{,\varrho}}{1+h} + 2(\gamma \nu_{,\varrho} - \xi_{,\varrho})\right)\left(A_{\varrho,z} - A_{z,\varrho}\right) + \partial_\varrho \left(A_{\varrho,z} - A_{z,\varrho}\right) &= 0, \label{tor_maxwell_2}
\end{align}
and the boundary conditions,
\begin{equation} \label{bc_infinity}
\lim_{|(\varrho,z)| \to \infty} (|\nu| + |\xi| + |\omega| + |h| + |A_t| + |A_\varrho| + |A_z| + |a|)(\varrho,z) = 0
\end{equation}
at spatial infinity and
\begin{align}
\xi(0,z) &= \ln(1+h(0,z)), \qquad z\in\mathbb R \label{bc_center} 
\end{align}
at the centre of symmetry.

\begin{rem}
The connection between equations (\ref{final_eq_nu})--(\ref{bc_center}) and the EVM-system is addressed in Proposition \ref{prop_equivalent} below.
\end{rem}

\begin{rem}  \label{rem_mag_rot}
If the ansatz function $f = \phi(E)\psi(\lambda, L)$ for the matter distribution satisfies in addition to the conditions listed on page \pageref{cond_phi} that $\psi$ is even in $L$, then the equations (\ref{final_eq_nu})--(\ref{bc_center}) possess solutions such that $\omega \equiv a \equiv 0$, i.e.~static solutions without rotation. Note that the corresponding matter functions vanish, cf.~(\ref{m4_m6_0}). \par
So the equations exhibit the physical connection between rotation and the magnetic field. Intuitively one would think of this connection in the following way. If there is no overall rotation, i.e.~$\omega \equiv 0$, then there is consequently no electric current and no magnetic field is induced. If there is rotation, however, the moving charges induce a poloidal magnetic field. Inspecting equations (\ref{eq_omega}) and (\ref{final_eq_a}), we see that $\omega\equiv a \equiv 0$ is a solution, whereas it is not possible that only one of these functions is zero everywhere because they appear mutually as source terms in the equation of each other. 
\end{rem}

\begin{lemma} \label{lem_no_tor}
For each continuous solution of (\ref{final_eq_nu})--(\ref{bc_center}) the combination $\beta = A_{\varrho,z} -  A_{z,\varrho}$ vanishes everywhere, i.e.~there is no toroidal magnetic field. (The toroidal component of the magnetic field is given in (\ref{tor_mag})).
\end{lemma}

\begin{proof}
If we consider the quantity $\beta = A_{\varrho,z} - A_{z,\varrho}$, then equations (\ref{tor_maxwell_1})--(\ref{tor_maxwell_2}) read $\nabla \beta = -\beta \nabla\left( \ln(\varrho(1+h)) + 2(\gamma \nu + \xi)\right)$. This admits the solution
\begin{equation}
\beta = C e^{-(\ln(\varrho(1+h))+2(\gamma\nu-\xi))}.
\end{equation}
Since $-(\ln(\varrho(1+h))+2(\gamma\nu-\xi)) \to \infty$, as $\varrho \to 0$ we deduce that $C=0$ since otherwise the toroidal component of the magnetic field would diverge as $\varrho \to 0$, hence the assumption of a regular $\{ \varrho = 0 \}$-axis would be violated.
\end{proof}

Taking account for the fact that the magnetic field is purely poloidal, we exclude the corresponding equations (\ref{tor_maxwell_1})--(\ref{tor_maxwell_2}) from our notion of the reduced EVM system, i.e. we make the following definition. 

\begin{defi} \label{def_red_evm}
The reduced EVM-system with parameters $\gamma$, $\lambda$ is defined as equations (\ref{final_eq_nu})--(\ref{final_eq_a}), equipped with the boundary conditions (\ref{bc_infinity})--(\ref{bc_center}).
\end{defi}

The axially symmetric solutions of the EVM-system which are constructed in this article are obtained as perturbations around spherically symmetric solutions of the Vlasov-Poisson system. For this reason we discuss the non-relativistic limit of the EVM-system, i.e.~the limit where $\gamma \to 0$. \par 
Define for the spherically symmetric steady state of the Vlasov Poisson system for particles of mass $1-q^2$ the potential at infinity $U_\infty$ by
\begin{equation}
U_\infty := \lim_{|x|\to \infty} U_N(x).
\end{equation}
Then by condition (2) on $\phi$ we clearly have $U_\infty > E_0$ and there exists $R \in (R_N, \infty)$ such that
\begin{equation} \label{def_cap_r}
U_N(r) > \frac{E_0 + U_\infty}{2}, \quad \mathrm{for\, all} \; r > R.
\end{equation}
(Recall that $R_N$ is such that $U_N(R_N)=E_0$.) It turns out, that in the limit $\gamma\to 0$, only the equations (\ref{final_eq_nu}) and (\ref{eq_a0}) of the reduced EVM-system remain non-trivial and they reduce to the Poisson equations
\begin{align}
\Delta \nu_N &= 4\pi\rho_{\nu_N + q A_N}, \label{newt_1}\\
\Delta A_N &= -4\pi q \rho_{\nu_N + q A_N}, \label{newt_2}
\end{align}
where we use the notation $\rho_{\nu_N + q A_N}$, introduced in (\ref{def_rho_u}), on the right hand side. See the proof of Lemma \ref{lem_n_zero} for details. \par
The system (\ref{newt_1})--(\ref{newt_2}) equipped with the boundary conditions
\begin{equation}
\nu_N(0)=0, \quad A_N(0)=0,
\end{equation}
and the equation
\begin{equation} \label{poisson_eq_un}
\Delta U_N = 4\pi (1-q^2) \rho_N, \quad U_N(0)=0
\end{equation}
are equivalent in the sense that a solution of (\ref{newt_1})--(\ref{newt_2}) gives rise to a solution of (\ref{poisson_eq_un}) via $U_N = \nu_N + q A_N$ and a solution of (\ref{poisson_eq_un}) gives rise to a solution of (\ref{newt_1})--(\ref{newt_2}) via $\nu_N = (1-q^2)^{-1} U_N$, $A_N = -q(1-q^2)^{-1} U_N$. In Lemma \ref{lem_sol_decay} below we will furthermore see that the limits $\nu_\infty = \lim_{|x|\to \infty} \nu$ and $A_\infty = \lim_{|x|\to \infty} A_t$ exist for any $\gamma\in(0,\infty)$ and that in the limit $\gamma\to 0$ there holds $A_\infty = -q\nu_\infty$, which is consistent. \par
 We are going to linearise around a solution of the system in the limit $(\gamma, \lambda)\to (0,0)$. We denote this solution by $\zeta_0$, i.e.
\begin{equation}
\zeta_0 = (\nu_N,0,0,0,A_N,0).
\end{equation}

\begin{lemma} \label{lem_comp_supp}
If $\gamma > 0$ is sufficiently small, then the matter quantities of a solution $\zeta$ of the reduced EVM-system are supported within a ball of radius $R$ around the origin.
\end{lemma}

\begin{proof}
The particle energy $E$ converges to the Newtonian particle energy $E_N$, given by
\begin{equation} \label{newtonian_energy}
E_N := \frac{|v|^2}{2} + \nu_N + \omega L_N +  qA_N, \quad L_N=\varrho v^3
\end{equation}
in the non-relativistic limit where $\gamma \to 0$. Using the expansions $e^x = 1+ x+ \dots$ and $\sqrt{1+x} = 1+\frac 12 x + \dots$ we obtain
\begin{align}
E &= \frac{e^{\gamma \nu} \sqrt{1+\gamma |v|^2}-1}{\gamma} + \omega\tilde L + qA_t \label{newtonian_l_e} \\
&= \frac{|v|^2}{2} + \nu + \left(\frac{\nu^2}{2} - \frac{|v|^4}{4} + \frac{\nu|v|^2}{2}\right)\gamma + \dots + \omega \tilde L  + qA_t 
\end{align}
and since $\nu \to \nu_N$, $A_t \to A_N$, $\omega \to 0$, we see $E \to E_N$ as $\gamma \to 0$. which is the Newtonian particle energy with potential $U_N = \nu_N + qA_N$. \par
Now, since $\|\nu + qA_t - U_N\|_\infty \to 0$, as $\gamma \to 0$, there is $\gamma_0 > 0$ such that for all $0 \leq \gamma \leq \gamma_0$ we have $E > \nu + qA_t > E_0$ for all $|x| > R$.
\end{proof}

\section{The function space of the solution} \label{sect_function_space}

In this paragraph the function spaces are defined in which a solution $\zeta = (\nu, h, \xi, \omega, A_t, a)$ of the reduced EVM-system will be constructed. In \cite{akr11, akr14} the considered function spaces contain axially symmetric functions on $\mathbb R^3$. Taking account for the fact that the reduced EVM-system is formulated as Poisson equations in different dimensions we define the function spaces for functions in the according dimensions. Furthermore, for the analysis of the source terms of these Poisson equations a hierarchy in regularity among the individual solution functions is needed, cf.~Lemma \ref{lem_reg} below. For this reason the assumed regularity is a bit stronger than in \cite{akr14}. \par
Let $\alpha\in(0, 1/2)$ be a fixed parameter and $Z_R = \{(x^1, x^2, x^3) \in \mathbb R^3 \, : \, \varrho(x^1,x^2) \leq R\}$. We define the following spaces of axially symmetric functions,
\begin{align}
\mathcal X_1 &:= \{ \nu \in C^{3,\alpha}(\mathbb R^3)\, |\, \nu =  \nu(\varrho, z) = \nu(\varrho, -z),\; \mathrm{and} \; \| \nu\|_{\mathcal X_1} < \infty\}, \\
\mathcal X_2 &:= \{ h \in C^{3,\alpha}(\mathbb R^4)\, |\,  h = h(\varrho, z) = h(\varrho, -z),\; \mathrm{and} \; \|  h \|_{\mathcal X_2} < \infty\}, \\
\mathcal X_3 &:= \{ \xi \in C^{1,\alpha}(Z_R)\, |\, \xi = \xi(\varrho,z) = \xi(\varrho, -z),\; \mathrm{and} \; \|  \xi \|_{\mathcal X_3} < \infty\}, \\
\mathcal X_4 &:= \{ \omega \in C^{2,\alpha}(\mathbb R^5)\, |\, \omega = \omega(\varrho, z) = \omega(\varrho, -z), \; \mathrm{and} \; \| \omega\|_{\mathcal X_4} < \infty\},
\end{align}
and
\begin{equation}
\mathcal X := \mathcal X_1 \times \mathcal X_2 \times \mathcal X_3  \times \mathcal X_4 \times \mathcal X_1 \times \mathcal X_4.
\end{equation}
Let $\beta \in (0,1)$ be another fixed parameter. Then the corresponding norms are defined to be
\begin{align}
\|  \nu \|_{\mathcal X_1} &:= \| \nu \|_{C^{3,\alpha}(\mathbb R^3)} + \left\|(1 + |x|)^{1+\beta} \nabla  \nu\right\|_\infty, \\
\|  h \|_{\mathcal X_2} &:= \| h \|_{C^{3,\alpha}(\mathbb R^4)} + \left\|(1+|x|)^3 \nabla h \right\|_\infty, \\
\|  \xi \|_{\mathcal X_3} &:= \| \xi \|_{C^{1,\alpha}(Z_R)}, \\
\|  \omega \|_{\mathcal X_4} &:= \| \omega \|_{C^{2,\alpha}(\mathbb R^5)} + \|(1+|x|)^3 \omega\|_\infty + \|(1+|x|)^4 \nabla  \omega\|_\infty, \label{def_norm_x4}
\end{align}
and
\begin{equation}
\|\zeta\|_{\mathcal X} := \| \nu \|_{\mathcal X_1} + \| h \|_{\mathcal X_2} + \| \xi \|_{\mathcal X_3} + \| \omega \|_{\mathcal X_4} + \| A_t \|_{\mathcal X_1} + \| a \|_{\mathcal X_4}.
\end{equation}
Finally we define
\begin{equation}
\mathcal U := \{(\zeta,p) \in \mathcal X \times [0,\delta) \times (-\delta,\delta))\, |\, \| \zeta -\zeta_0) \|_{\mathcal X} < \delta_0\},
\end{equation}
where $\delta_0 > 0$ is sufficiently small such that for all $(\zeta;\gamma,\lambda) \in \mathcal U$, we have $1+h(\varrho,z) > 1/2$ for all $(\varrho,z) \in [0,\infty) \times \mathbb R$.

\section{Solutions of the reduced system solve the full EVM-system}

In this article we construct solutions to the reduced EVM-system (\ref{final_eq_nu})--(\ref{bc_center}). These solutions to the reduced EVM-system correspond to spherically symmetric, time independent solutions of the EVM-system (\ref{eq_einstein})--(\ref{maxwell_eq_2}). The relations between these systems is the subject of the following proposition. As already mentioned, this article generalises \cite{akr14} to the case of charged particles and the reduced system treated here coincides with the reduced system considered in \cite{akr14} if the charge parameter $q$ is set to zero.

\begin{prop} \label{prop_equivalent}
A solution $\zeta \in \mathcal X$ of the reduced EVM-system (\ref{final_eq_nu})--(\ref{bc_center}) with parameters $\lambda$, $\gamma$ gives rise to a time independent, axially symmetric solution $(g,f,A)$ of the EVM-system (\ref{eq_einstein})--(\ref{maxwell_eq_2}) where $g$ is of the form (\ref{ansatz_metric}) and $f$ is of the form (\ref{ansatz_f}).
\end{prop}

Before we prove Proposition \ref{prop_equivalent} we establish the following scaling law.

\begin{lemma} (Scaling law) \label{lem_scaling} \\
Let $(\nu, h, \xi, \omega, f, A_t, a)$ be a solution of the reduced EVM-system (\ref{eq_einstein})--(\ref{maxwell_eq_2}) with parameters $(\lambda, c) \in (-1,1) \times (0,\infty)$. Then the functions $\tilde \nu, \tilde h, \tilde \xi, \tilde \omega, \tilde f, \tilde A_t, \tilde a$, given by
\begin{multline}
\left(\tilde \nu(\varrho, z), \tilde h(\varrho, z), \tilde \xi(\varrho, z), \tilde \omega(\varrho, z), \tilde A_t(\varrho, z), \tilde a(\varrho, z)\right) \\ = \left(\frac{1}{c^2} \nu(c\varrho, cz), h(c\varrho, cz), \xi(c\varrho, cz), \omega(c\varrho, cz), \frac{1}{c^2} A_t(c\varrho, cz), a(c\varrho, cz)\right)
\end{multline}
and
\begin{equation} \label{scaling_f}
\tilde f(\varrho, z, p^\varrho, p^z, p^\varphi) = c^3 f(c\varrho, cz, cp^\varrho, cp^z, p^\varphi)
\end{equation}
satisfy the reduced EVM-system with parameters $(\lambda, 1)$. 
\end{lemma}

\begin{proof}
We check the laws for $A_t$ and $a$. For the other functions, cf.~\cite{akr14}. For the Laplace operator we have the transformation law
\begin{equation}
\Delta \tilde A_t (\varrho, z) = (\Delta A_t)(c\varrho, cz).
\end{equation}
Then we use the Maxwell equations (\ref{eq_a0}) and (\ref{final_eq_a}) for $A_t$ and $a$, respectively. Note that for example
\begin{equation}
\left(\nabla A_t\right)(c\varrho, cz) = \frac{1}{c} \nabla (A_t(c\varrho, cz)) = c \nabla \tilde A_t(\varrho, z).
\end{equation}
For the matter function corresponding to $A_t$ we obtain the expression
\begin{align}
&\mathfrak M_5[\zeta; \gamma, \lambda](c\varrho, cz) \\
&= -4\pi q e^{(2\tilde \xi-3 \tilde \nu)(\varrho,z)} \int_{\mathbb R_v^3} f\left(c\varrho, cz, p^\varrho(c\varrho, cz, v^1), p^z (c\varrho, cz, v^2), p^\varphi(c\varrho, cz, v^3)\right) \nonumber \\
&\hspace{5cm} \times \left(e^{2\tilde\nu(\varrho,z)} + \frac{\varrho(1+\tilde h(\varrho,z)\omega(\varrho,z) v^3}{c\sqrt{1+\gamma |v|^2}}\right) \, \mathrm dv^1 \mathrm dv^2 \mathrm dv^3 \nonumber
\end{align}
and for $a$ we have the matter function
\begin{multline}
\mathfrak M_6[\zeta; \gamma, \lambda](c\varrho, cz) = \frac{4\pi q}{c}  \varrho (1 + \tilde h(\varrho,z)) e^{(2\tilde \xi- 3 \tilde \nu)(\varrho,z)} \\\times \int_{\mathbb R_v^3} f\left(c\varrho, cz, p^\varrho(c\varrho, cz, v^1), p^z (c\varrho, cz, v^2), p^\varphi(c\varrho, cz, v^3)\right) \frac{v^3}{\sqrt{1+\gamma |v|^2}}\, \mathrm dv^1 \mathrm dv^2 \mathrm dv^3.
\end{multline}
Now, applying the change of variables $v^i \to w^i = v^i/c$, $i=1,2,3$, and using the scaling law (\ref{scaling_f}) one recovers the original matter functions with $\tilde f$ instead of $f$.
\end{proof}

\begin{proof}[Proof of Proposition \ref{prop_equivalent}]
First we describe how the reduced EVM-system can be derived from the EVM-system. We start with the equations (\ref{final_eq_nu})--(\ref{eq_omega}) which --without electro-magnetic field terms of course-- have been considered in \cite{akr14}. Write down all Einstein equations in the coordinates $t,\varrho, \varphi, z$ and take into account the symmetries by substituting the ansatz (\ref{ansatz_metric}) for $g$. Suitable combinations of the Einstein equations yield the equations (\ref{final_eq_nu})--(\ref{eq_omega}) for $\nu$, $h$, $\xi$, and $\omega$. For equation (\ref{final_eq_nu}) take the combination
\begin{equation} \label{combination_nu}
\frac 12 \left(e^{2\xi - 4\gamma\nu}(G_{tt} + 2\omega G_{t\varphi}) + \frac{1}{\gamma} (G_{\varrho\varrho} + G_{zz}) + e^{2\xi}\left(\frac{1}{\gamma \varrho^2 (1+h)^2} + \omega^2 e^{-4\gamma\nu}\right) G_{\varphi\varphi} \right).
\end{equation}
For equation (\ref{eq_b}) take $(1+h)(G_{\varrho\varrho} + G_{zz})$, for equation (\ref{eq_omega}) take $\frac{2 e^{2\xi}}{\varrho^2 (1+h)^2} (G_{t\varphi} + \omega G_{\varphi\varphi})$, and for equation (\ref{eq_xi}) take
\begin{equation} \label{combination_xi}
(1+h+\varrho h_{,\varrho}) \frac{(1+h)\varrho}{2} (G_{\varrho\varrho} - G_{zz}) + \varrho^2 h_z (1+h) G_{\varrho z}.
\end{equation}
It is important to take the right combination of Einstein equations for the method to work and we follow \cite{akr11}. \par
The components $G_{\mu\nu}$ of the Einstein tensor and the components $\tau_{\mu\nu}$ of the electro-magnetic part of the energy momentum tensor yield the left members and the source functions of equations (\ref{final_eq_nu})--(\ref{eq_omega}). The matter functions $M_i^{(\gamma,\lambda)}$, $i=1,2,4$ are obtained as explained now. First, using the ansatz (\ref{ansatz_f}) for the particle distribution function $f$ and the orthonormal frame (\ref{def_frame}) one can write the components of the kinetic part $T_{\mu\nu}$ of the energy momentum tensor, defined in (\ref{eq_em_tensor}), as the integral expression.
\begin{equation}
T_{\mu\nu} = \int_{\mathbb R_v^3} \phi(E) \psi(\lambda, L) \frac{p_\mu p_\nu}{\sqrt{1+\gamma |v|^2}}\, \mathrm dv^1 \mathrm dv^2 \mathrm dv^3. \label{em_tensor_v}
\end{equation}
For this formula the mass shell relation (\ref{mass_shell_frame}) needs to be used. Furthermore, the variables $p_\mu$, $\mu=0,\dots, 3$ can in terms of the frame components $v^1$, $v^2$, $v^3$, be expressed as
\begin{equation} \label{pd_via_v}
\begin{aligned}
p_0 &= -\frac{e^{\gamma\nu}}{\gamma}\sqrt{1+\gamma |v|^2} - e^{-\gamma\nu}\varrho (1+h) \omega v^3, \\
p_\varrho &= e^\mu v^1, \quad p_z = e^\mu v^2, \quad p_\varphi = e^{-\gamma\nu} \varrho (1+h) v^3.
\end{aligned}
\end{equation}
Now taking the corresponding combinations of $T_{\mu\nu}$ and substituting the expressions (\ref{pd_via_v}) for the $p$-variables one obtains after simplification the matter functions. These matter functions coincide with the corresponding matter terms in \cite{akr14}, the only difference consists in the quantities $E$ and $L$. The matter quantity $M_3^{(\gamma,\lambda)}$ vanishes due to the symmetry $T_{\varrho\varrho} = T_{zz}$. \par
The equations (\ref{eq_a0}) and (\ref{final_eq_a}) for $A_t$ and $a$, respectively, are new with respect to \cite{akr14} and they are obtained by suitable combinations of the Maxwell equation $\nabla_\alpha F^{\alpha \beta} = -4\pi q J^\beta$ for $\beta = t$ and $\beta = \varphi$. These combinations are 
\begin{align}
&\frac{1}{\gamma} e^{2\xi} \nabla_\alpha F^{\alpha t} - \omega \varrho^2 (1+h)^2 e^{2\xi-4\gamma\nu}\left(\omega \nabla_\alpha F^{\alpha t} - \nabla_\alpha F^{\alpha\varphi}\right), \label{combination_a0} \\
&(1+h)^2 e^{2\xi-4\gamma\nu}\left(\omega \nabla_\alpha F^{\alpha t} - \nabla_\alpha F^{\alpha\varphi}\right), \label{combination_a}
\end{align}
respectively. The matter functions $M_5^{(\gamma,\lambda)}$ and $M_6^{(\gamma,\lambda)}$ are obtained by taking the respective combinations of the components of the matter current $J^\beta$, defined in (\ref{maxwell_eq_2}). Using the orthonormal frame (\ref{def_frame}) it can be written as
\begin{equation}
J^\beta = \gamma \int_{\mathbb R_v^3} \phi(E) \psi(\lambda, L) \frac{p^\beta}{\sqrt{1+\gamma |v|^2}}\, \mathrm dv^1 \mathrm dv^2 \mathrm dv^3. \label{matter_current_v}
\end{equation}
The variables $p^\mu$, $\mu= 0,\dots,3$, are given in terms of the frame coordinates as
\begin{equation} \label{frame_reverse}
p^0 = e^{-\gamma\nu} v^0, \quad p^1 = e^{-\mu} v^1, \quad p^2 = e^{-\mu} v^2, \quad p^3 = e^{-\gamma\nu} \omega v^0 + \frac{e^{\gamma\nu}}{(1+h)\varrho} v^3. 
\end{equation}
So far it has been proved that a solution of the EVM-system implies a solution of the reduced EVM-system since the latter one is obtained by linear combinations of certain components of the former one. It remains to verify that the converse is also true, i.e.~that a solution to the reduced EVM-system with parameter $c \in [1,\infty)$ implies an axially symmetric, time independent solution of the EVM-system with $c=1$. First we note that by the scaling laws (Lemma \ref{lem_scaling}) a solution to the reduced EVM-system with $c=1$ can always be obtained. The Maxwell equations are already fulfilled since the number of equations has not been reduced. For the Einstein equations however the number of equations has been reduced, so situation is less clear. We define the quantity
\begin{equation}
E_{\mu\nu} = G_{\mu\nu} - \frac{8\pi}{c^4} \left(T_{\mu\nu} + \tau_{\mu\nu}\right), \quad \mu, \nu = t, \varrho, z, \varphi.
\end{equation}
The non-trivial components are $E_{tt}$, $E_{\varrho\varrho}$, $E_{zz}$, $E_{\varphi \varphi}$, $E_{t\varphi}$, and $E_{\varrho z}$. The other components are trivially zero since the Einstein tensor vanishes under the symmetry assumptions incorporated into the metric ansatz (\ref{ansatz_metric}). It remains to show that the components $E_{tt}$, $E_{\varrho\varrho}$, $E_{zz}$, $E_{\varphi \varphi}$, $E_{t\varphi}$, and $E_{\varrho z}$ vanish, too. This can be done by using the same argument as given in \cite[Section 6]{akr14} since the Einstein part of the reduced EVM-system that we are working with consists in the same linear combinations of Einstein equations which has been considered in \cite{akr14}. A subtlety, which has to be dealt with, consists in the fact that $\xi$ is only $C^{1,\alpha}$, whereas Einstein's equations are of second order. Since in the present setup $\xi$ has the same regularity as in the setup of \cite{akr14} the arguments of \cite{akr14} apply however. \par
Finally, the boundary conditions (\ref{bc_infinity}) clearly imply the boundary conditions (\ref{bc_as_flat}).
\end{proof}

\section{Definition of the solution operator $\mathfrak F$} \label{sect_def_f}

The equations (\ref{final_eq_nu}), (\ref{eq_b}), (\ref{eq_omega})--(\ref{final_eq_a}) of the reduced EVM-system are semi-linear Poisson equations. For this reason the solution operators corresponding to these equations are basically given in terms of the Greens function of the Laplace operator. If $q$ is set to zero, the solution operator introduced here coincides with the solution operator defined in \cite{akr14}. \par
First, we recall some facts about the Poisson equation. Define for $n\geq 3$ the $n$-dimensional Greens function $G^n_y(x)$ of the Laplace operator $\Delta_n$ by
\begin{equation} \label{def_g_y}
G^n_y(x) = \frac{1}{(n-2)|\mathbb S^{n-1}|} \frac{1}{|x-y|^{n-2}},
\end{equation}
where $|\mathbb S^{n-1}|$ is the volume of the $(n-1)$-dimensional unit sphere. For later convenience we also define
\begin{equation}
\hat G^n_y(x) = \frac{1}{(n-2)|\mathbb S^{n-1}|} \left( \frac{1}{|x-y|^{n-2}} - \frac{1}{|y|^{n-2}}\right)
\end{equation}
and the functionals
\begin{equation} \label{def_greens_gn}
G_n[f](x) := \int_{\mathbb R^n} G^n_y(x) f(y)\, \mathrm dy \quad \mathrm{and} \quad \hat G_n[f](x) := \int_{\mathbb R^n} \hat G^n_y(x) f(y)\, \mathrm dy.
\end{equation}
Then, in the sense of distributions, the solution of the Poisson equation $-\Delta_n u = f$ for $f\in L_{\mathrm{loc}}^1(\mathbb R^n)$ on $\mathbb R^n$, $n\geq 1$ is given by $u(x) = G_n[f](x)$, cf.~\cite[Theorem 6.21]{ll00}. \par
Now we give the definition of $\mathfrak F$. To this end we first define the operators $\mathfrak G_i :\mathcal U \to \mathcal X_i$, $i=1,\dots,6$ (by $\mathcal X_5$ and $\mathcal X_6$ we understand $\mathcal X_1$ and $\mathcal X_4$, respectively). We define
\begin{align}
\mathfrak G_i[\zeta; \gamma, \lambda] &:= G_3[\mathfrak g_i[\zeta; \gamma]] + \hat G_3[\mathfrak M_i[\zeta; \gamma, \lambda]], &i &= 1,5,  \label{def_g15}\\
\mathfrak G_2[\zeta; \gamma, \lambda] &:= G_4[\mathfrak M_2[\zeta; \gamma, \lambda]], && \label{def_g2}\\
\mathfrak G_3[\zeta; \gamma, \lambda] &:= \ln(1+h(0,z)) + \int_0^\varrho \mathfrak g_3[\zeta; \gamma](s,z) \, \mathrm ds, &&\label{def_g3} \\
\mathfrak G_i[\zeta; \gamma, \lambda] &:= G_5[\mathfrak g_i[\zeta; \gamma] + \mathfrak M_i[\zeta; \gamma, \lambda]], &i &= 4,6. \label{def_g46}
\end{align}
Then we write compactly
\begin{equation}
\mathfrak G[\zeta; \gamma, \lambda] := (\mathfrak G_1[\zeta; \gamma, \lambda], \dots, \mathfrak G_6[\zeta; \gamma, \lambda]).
\end{equation}
Furthermore we define
\begin{equation}
\mathfrak F:\mathcal U \to \mathcal X,\quad (\zeta; \gamma, \lambda) \mapsto \mathfrak F[\zeta; \gamma, \lambda] := \zeta - \mathfrak  G[\zeta; \gamma, \lambda].
\end{equation}

\begin{lemma} \label{lem_rel_greens}
Let $(\zeta; \gamma, \lambda) \in \mathcal U$. Then $\mathfrak G_i[\zeta; \gamma, \lambda]$ is axially symmetric and even in the $x^n$-coordinate (also referred to as $z$-coordinate) for all $i=1,\dots,6$.
\end{lemma}

\begin{proof}
Clearly $\mathfrak g_i[\zeta; \gamma]$ and $\mathfrak M_i[\zeta; \gamma, \lambda]$ are axially symmetric and even in $z$ if $\zeta$ is. Consider the following prototype term. Let $f:\mathbb R^n \to \mathbb R^n$ be an axially symmetric function that is even in $x^n=z$. One can check straight forwardly that $G_n[f]$ is axially symmetric and even in $z$ by performing and appropriate change of variables in the integral, i.e.~we have for $A\in SO(n-1)$
$$
G_n[f](A \cdot (x^1, \dots, x^{n-1})^\intercal, -x^n) = G_n[f](x).
$$
\end{proof}

\begin{rem}
The operators $\mathfrak G_1$ and $\mathfrak G_5$ have been defined such that the Fr\'echet derivative of $\mathfrak F$ with respect to $\nu$, $A_t$, at $(\zeta_0;0,0)$ is zero at $(\varrho,z)=0$. Observe the $\hat G$ in equation (\ref{def_g15}). This property is important in the proof that the Fr\'echet derivative at $(\zeta_0; 0, 0)$ is a bijection, cf.~Lemma \ref{lem_bijection} below.
\end{rem}

\begin{prop} \label{prop_consistent}
Let $\zeta\in\mathcal X$ and $(\gamma,\lambda) \in [0,\delta) \times (-\delta,\delta)$. Then $\mathfrak F[\zeta; \gamma, \lambda] = 0$ if and only if $\zeta$ restricted to $\{\varrho \geq 0\}$ is a solution of the reduced EVM-system (\ref{final_eq_nu})--(\ref{final_eq_a}) with parameters $\gamma$, $\lambda$.
\end{prop}

\begin{proof}
The statement is clear for $\mathfrak G_i$ and $\zeta_i$, $i=1,2,4,5,6$ since by Lemma \ref{lem_rel_greens} these operators are the solution operators to the semi-linear Poisson equations (\ref{final_eq_nu}), (\ref{eq_b}), (\ref{eq_omega})--(\ref{final_eq_a}). For the operator $\mathfrak G_3$, we observe that differentiation of $\mathfrak G_3[\zeta;\gamma,\lambda](\varrho,z)$ with respect to $\varrho$ directly yields the right hand side of the $\xi$-equation (\ref{eq_xi}).
\end{proof}

\begin{lemma} \label{lem_n_zero}
Recall $\zeta_0 = (\nu_N,0,0,0,A_N,0)$. We have $\mathfrak F[\zeta_0;0,0] = 0$.
\end{lemma}

\begin{proof}
We adopt the notation $\rho_N := \rho_{U_N}$, $\alpha_N := \alpha_{U_N}$. The Einstein equations (\ref{eq_b})--(\ref{eq_omega}) for $h$,  $\xi$, and $\omega$ are trivially satisfied for $\zeta=\zeta_0$. So it remains to consider equation (\ref{final_eq_nu}) for $\nu$. The source function $\mathfrak g_1[\zeta_0; 0, 0]$ is zero. For the matter function $\mathfrak M_1$ a calculation yields $\mathfrak M_1[\zeta_0; 0, 0](\varrho,z) = 4\pi \rho_N(r)$, where $r=\sqrt{\varrho^2 + z^2}$. This is the energy density induced by the ansatz (\ref{ansatz_f}) in the Newtonian case. \par 
We see that the Maxwell equation (\ref{final_eq_a}) for $a$ is satisfied with $\gamma=0$ and $a\equiv \omega\equiv h \equiv 0$. Concerning the Maxwell equation (\ref{eq_a0}) for $A_t$, we see that it reduces to
\begin{equation}
\Delta_3 A_t = -4\pi q \rho_N(r).
\end{equation}
 So $U_N = \nu_N + qA_N$ solves the Poisson equation
\begin{equation}
\Delta U_N(r) = 4\pi (1-q^2) \rho_N(r).
\end{equation}
Note also that we are using the assumption $\psi(0,L)=1$. So we actually obtain
\begin{equation}
U_N(r) = \mathfrak G_1[\zeta_0;0,0](\varrho,z) + q \mathfrak G_5[\zeta_0;0,0](\varrho,z)
\end{equation}
and the assertion follows.
\end{proof}

\section{$\mathfrak F$ is well defined} \label{sect_well_defined}

We have to verify that for all $(\zeta; \gamma, \lambda) \in \mathcal U$ the functions $\mathfrak G_i[\zeta; \gamma, \lambda]$ satisfy the regularity conditions and the decay behaviour stated in the definition of $\mathcal X$, for $i=1,\dots,6$. \par

Before we prove the regularity properties of $\mathfrak G[\zeta;  \gamma, \lambda]$ we collect a few facts on axially symmetric functions, proven in \cite{akr11} and \cite{akr14}.

\begin{lemma} (Lemma 7.1 in \cite{akr14}) \label{lem_axially} \\
Let $u:\mathbb R^n \to \mathbb R$ be axially symmetric and $u(x) = \tilde u(\varrho,z)$ where $\tilde u:[0,\infty)\times \mathbb R \to \mathbb R$. Let $k\in \{1,2,3\}$ and $\alpha \in (0,1)$. Then
\begin{enumerate}
\item $u\in C^k(\mathbb R^n) \Leftrightarrow \tilde u \in C^k([0,\infty) \times \mathbb R)$ and all derivatives of $\tilde u$ of order up to $k$ which are of odd order in $\varrho$ vanish for $\varrho = 0$,
\item $u\in C^{0,\alpha}(\mathbb R^n) \Leftrightarrow \tilde u \in C^{0,\alpha}([0,\infty) \times \mathbb R)$.
\end{enumerate}
\end{lemma}

\begin{lemma} (Lemma 3.2 in \cite{akr11}) \label{lem_axially_2}\\
Let $\varphi=\varphi(\varrho,z)\in C^4(\mathbb R^2)$ be odd in $\varrho$ and define
\begin{equation}
\zeta(\varrho,z) := \left\{\begin{array}{ll} \varphi(\varrho,z)/\varrho, & \varrho \neq 0, \\ \partial_\varrho \varphi(0,z), & \varrho = 0.\end{array}\right.
\end{equation}
Then $\zeta\in C^3(\mathbb R^2)$ and all derivatives of $\zeta$ up to order $3$ which are of odd oder in $\varrho$ vanish for $\varrho = 0$. By abuse of notation, $\zeta\in C^3(\mathbb R^3)$.
\end{lemma}


Next we establish regularity of the matter functions.

\begin{lemma}
Let $(\zeta; \gamma, \lambda) \in \mathcal U$. Then the functions $\mathfrak M_i[\zeta; \gamma, \lambda]$, $i=1,2,4,5,6$, if extended to negative values of $\varrho$ and thus seen as functions on $\mathbb R^2$, are even in $\varrho$.
\end{lemma}

\begin{proof}
That the matter functions are even in $\varrho$ has already been observed in \cite{akr14} and the new matter functions $M_5^{(\gamma,\lambda)}$ and $M_6^{(\gamma,\lambda)}$ can be treated with the same ideas. We perform in the integrals of the formulas (\ref{def_hat_m_1})--(\ref{def_hat_m_6}) for the matter functions $M_i^{(\gamma, \lambda)}$, $i = 1,2,4,5,6$, a change of variables, given by 
\begin{equation}
\eta = \frac{e^{\gamma\nu}\sqrt{1+\gamma |v|^2} -1}{\gamma}, \quad s = (1+h) e^{-\gamma\nu} v^3.
\end{equation}
Let
\begin{equation}
m(\eta, h, \nu) := (1+h)e^{-\gamma \nu} \sqrt{\frac{e^{-2\gamma \nu}(\gamma \eta + 1)^2-1}{\gamma}}.
\end{equation}
Then the domain of integration can be parameterised by $\eta \in ((e^{\gamma\nu}-1)/\gamma,\infty)$, $s\in(-m,m)$. Further, for a function $g = g(s, \eta, h, \nu, \varrho\omega)$, which will be chosen among the choices
$$
1 + 4\gamma\eta+2\gamma^2 \eta^2, \quad m^2-s^2, \quad s(1+\gamma\eta), \quad s, \quad 1+\gamma\eta + \gamma\omega\varrho s,
$$
we define $M_{(\gamma, \lambda)}$ to be the operator which assigns to $g$ the function
\begin{equation}
\begin{aligned} 
&M_{(\gamma, \lambda)}[g] : \mathbb R^2 \times \left(- \frac 12, \infty\right) \times \mathbb R^3 \to \mathbb R, \\
&(\varrho, \nu, h, \omega, A_t, a) \mapsto M_{(\gamma, \lambda)}[g](\varrho, \nu, h, \omega, A_t, a) \label{def_capital_m} \\
&\hspace{1cm} = \int_{\frac{e^{\gamma\nu} - 1}{\gamma}}^\infty \int_{-m(\eta, h, \nu)}^{m(\eta, h, \nu)} \phi(\eta + \varrho\omega s + q A_t) \psi(\lambda, \varrho s - q \varrho^2 a)\, g(s, \eta, h, \nu, \varrho\omega)\, \mathrm ds \mathrm d\eta.
\end{aligned}
\end{equation}
The range $(-1/2,\infty)$ of $h$ is motivated by the definition of the set $\mathcal U$ of functions that we consider. Then the matter functions $M_i^{(\gamma,\lambda)}$ can be written in the form
\begin{align}
M_1^{(\gamma,\lambda)}(\varrho, \zeta) &= \frac{8\pi^2}{1+h} e^{2\xi - 4\gamma\nu} M_{(\gamma, \lambda)}\left[ 2(1+\gamma\eta)^2 - e^{2\gamma\nu} \right](\varrho, \nu, h, A_t, a), \label{m1_mg} \\
M_2^{(\gamma,\lambda)}(\varrho, \zeta) &= \frac{16 \pi^2 \gamma^2}{(1+h)^2} e^{2\xi}  M_{(\gamma, \lambda)}\left[m^2 - s^2\right](\varrho, \nu, h, A_t, a), \label{m2_mg}\\
M_4^{(\gamma,\lambda)}(\varrho, \zeta) &= -\frac{32 \pi^2 \gamma}{\varrho (1+h)^3}e^{2\xi} M_{(\gamma, \lambda)}[s(1 + \gamma\eta)](\varrho, \nu, h, A_t, a), \label{m4_mg}  \\
M_5^{(\gamma,\lambda)}(\varrho, \zeta) &= \frac{8\pi^2 q}{1+h} e^{2(\xi-\gamma\nu)} M_{(\gamma, \lambda)}\left[1 + \gamma\eta + \gamma\omega \varrho s\right](\varrho, \nu, h, A_t, a), \label{m5_mg} \\
M_6^{(\gamma,\lambda)}(\varrho, \zeta) &= -\frac{8\pi^2 q \gamma}{\varrho(1+h)} e^{2(\xi-\gamma\nu)} M_{(\gamma, \lambda)}[s](\varrho, \nu, h, A_t, a). \label{m6_mg}
\end{align}
Given these representations (\ref{m1_mg})--(\ref{m6_mg}) of the matter functions we observe the following fact. If $g(s, \eta, h, \nu, \varrho\omega)$ is even or odd in $s$ then $M_{(\gamma,\lambda)}[g](\varrho, \nu, h, A_t,  a)$ is even or odd in $\varrho$, respectively. To see this we substitute $-\varrho$ for $\varrho$ in the formula (\ref{def_capital_m}) for $M_{(\gamma,\lambda)}[g](\varrho, \nu, h, \omega, A_t,  a)$ and make then the change of variables $s \to \hat s = -s$. If $g$ is even in $s$ we obtain the same expression as for ``$+\varrho$'', whereas if $g$ is odd in $s$ we obtain its negative. \par
Then we observe that $\mathfrak M_i[\zeta; \gamma, \lambda]$ is even in $\varrho$ for all $i\in\{1,2,4,5,6\}$. Consider for example $\mathfrak M_5[\zeta; \gamma, \lambda]$, given by
\begin{align*}
&\mathfrak M_5[\zeta; \gamma, \lambda](\varrho, z) \\
&\qquad = \frac{8\pi^2 q}{1+h(\varrho, z)} e^{2(\xi-\gamma\nu)(\varrho, z)} M_{(\gamma, \lambda)}\left[1 + \gamma\eta\right](\varrho, \nu(\varrho, z), h(\varrho, z), A_t(\varrho, z), a(\varrho, z)) \\
&\qquad \quad + \frac{8\pi^2 q \gamma \omega(\varrho, z)  \varrho }{1+h(\varrho, z)} e^{2(\xi-\gamma\nu)(\varrho, z)}M_{(\gamma, \lambda)}\left[s\right](\varrho, \nu(\varrho, z), h(\varrho, z), A_t(\varrho, z), a(\varrho, z)).
\end{align*}
Here we view $\zeta \in \mathcal X$ as even functions in $\varrho$, cf.~Remark \ref{rem_even}. By the observation on $M_{(\gamma, \lambda)}$ which is mentioned above the first term is a product of functions that are even in $\varrho$. For the second term we observe that the fraction is odd in $\varrho$ since it contains $\varrho$ as explicit factor. The second factor is also odd in $\varrho$ by the upper observation. So in total the second term is even in $\varrho$.
\end{proof}

\begin{lemma} (Regularity of the matter functions) \label{lem_reg_m} \\
Let $\phi \in C_c^\kappa(\mathbb R), \psi\in C^\infty_c(\mathbb R^2)$, and $\gamma \in [0,1]$, $\lambda\in [-1/2,1/2]$, where $\kappa \geq 1$. Further, let $g\in C^\sigma(\mathbb R^5)$, for $\sigma \geq 1$. Then all partial derivatives up to order $\min\{\kappa+1,\sigma\}$ of the function $M_{(\gamma, \lambda)}[g]$, defined in (\ref{def_capital_m}), exist and are continuous. Furthermore, if
\begin{equation}
g(s, \eta, h, \nu,\varrho\omega) |_{\eta = l(s, \nu, h)} = 0,
\end{equation}
where $ l(s, \nu, b)$ is defined as
\begin{equation} \label{def_little_l}
l(s, \nu, h) := \frac 1 \gamma \left(e^{\gamma \nu} \sqrt{1+\gamma\frac{s^2 e^{2 \gamma \nu}}{(1+h)^2}}-1\right),
\end{equation}
then all partial derivatives up to order $\min\{\kappa+2,\sigma\}$ of $M_{(\gamma, \lambda)}[g]$ exist and are continuous.
\end{lemma}

\begin{proof}
We write down the integral representation (\ref{def_capital_m}) of $M_{(\gamma, \lambda)}[g]$ with respect to the new integration variable $\hat \eta := \eta + \varrho \omega + q A_t$. We obtain
\begin{multline} \label{exp_cm_1}
M_{(\gamma, \lambda)}[g](\varrho, \nu, h, \omega, A_t, a) \\= \int_{-\infty}^\infty \int_{l(s, \nu, h)+\varrho\omega + q A_t}^\infty \phi(\hat\eta)\psi(\lambda, \varrho s- q \varrho^2 a) g(s, \hat\eta - \varrho\omega- q A_t, h, \nu, \varrho\omega)\, \mathrm d\hat \eta \mathrm ds.
\end{multline}
We write this in a schematic form in order to make the analysis clearer. Let $x=(x_1, \dots, x_6)$. In the following this vector represents $(\varrho, \nu, h, \omega, A_t, a)$. We write
\begin{equation} \label{exp_cm_2}
M_{(\gamma, \lambda)}[g](x) = \int_{-\infty}^\infty \int_{\ell(s, x)}^\infty \phi(\hat\eta) \, \hat \psi(s, x) \, \hat g(s, \hat \eta, x)\, \mathrm d\hat\eta \mathrm ds,
\end{equation}
Where $\ell$, $\hat \psi$, and $\hat g$ are defined in the obvious way such that the expressions (\ref{exp_cm_1}) and (\ref{exp_cm_2}) agree, i.e.
\begin{align}
\ell(s,x) &= l(s, x_2, x_3) + x_1 x_4 + q x_5, \\
\hat \psi(s, x) &= \psi(\lambda, x_1 s - q x_1^2 x_ 6), \\
\hat g(s, \hat \eta, x) &= g(s, \hat \eta - x_1 x_4 - q x_5, x_3, x_2, x_1x_4).
\end{align}
Note that $\ell \in C^\infty(\mathbb R^3)$, since $l \in C^\infty(\mathbb R^3)$ already. To see the latter remind that $h > -1/2$ is assumed on the domain of $M_{(\gamma, \lambda)}$.\par
We have for $i=1,\dots,6$
\begin{align}
\partial_{x_i} M_{(\gamma, \lambda)}[g](x) &= \int_{-\infty}^\infty \int_{\ell(s,x)}^\infty \phi(\hat\eta) \,  \partial_{x_i}\left(\hat \psi(s, x) \hat g(s, \hat \eta, x)\right)\, \mathrm d\hat\eta \mathrm ds \\
&\quad + \int_{-\infty}^\infty \phi(\ell(s,x)) \, \hat \psi(s, x) \, \hat g(s, \ell(s,x), x)\, \partial_{x_i} \ell(s,x)\, \mathrm ds. \nonumber
\end{align}
Now we see that each additional derivative $\partial_{x_j}$, $j=1,\dots,6$ leads to a derivative acting on $\phi$, unless $\hat g(s, \ell(s,x), x) = 0$. In this case, only if there are three or more derivatives, there act one or more derivatives on $\phi$. Since $\hat\psi, \ell\in C^\infty$, and $\phi$ and $\psi$ are compactly supported, the regularity of $\phi$ and $g$ determines the regularity of $M_{(\gamma, \lambda)}[g]$ in the asserted way.
\end{proof}

Now we check the regularity properties of $\mathfrak G[\zeta; \gamma, \lambda]$.

\begin{lemma} \label{lem_reg}
Let $(\zeta; \gamma, \lambda)\in \mathcal U$. Then we have $\mathfrak G_1[\zeta; \gamma, \lambda], \mathfrak G_2[\zeta; \gamma, \lambda], \mathfrak G_5[\zeta; \gamma, \lambda]  \in C^{3,\alpha}(\mathbb R^2)$, $\mathfrak G_3[\zeta; \gamma, \lambda] \in C^{1,\alpha}(Z_R)$, and $\mathfrak G_4[\zeta; \gamma, \lambda], \mathfrak G_6[\zeta; \gamma, \lambda] \in C^{2,\alpha}(\mathbb R^2)$.
\end{lemma}

\begin{proof}
By \cite[Theorem 10.3]{ll00} the regularity of the axially symmetric solution functions $\mathfrak G_i[\zeta; \gamma, \lambda]$, $i=1,2,4,5,6$ follows from the regularity of the right members of the semi-linear Poisson equations (\ref{final_eq_nu}), (\ref{eq_b}), (\ref{eq_omega})--(\ref{final_eq_a}). These right members consist in the source functions $\mathfrak g_i[\zeta; \gamma]$ and the matter functions $\mathfrak M_i[\zeta; \gamma, \lambda]$. This regularity is now established. \par
We have already observed that all matter functions $\mathfrak M_j[\zeta; \gamma, \lambda]$, $j=1,2,4,5,6$ and all source functions $\mathfrak g_i[\zeta; \gamma, \lambda]$, $i = 1,4,5,6$, if extended to negative values of $\varrho$ and thereby seen as functions on $\mathbb R^2$, are even in $\varrho$ and $z$. So by Lemma \ref{lem_axially} it suffices to establish the necessary regularity in $\varrho$ and $z$. We start by analysing the matter functions $\mathfrak M_j[\zeta;  \gamma, \lambda]$, $j \in \{1,2,4,5,6\}$. By inspection of the formulas (\ref{m1_mg}), (\ref{m2_mg}), and (\ref{m5_mg}) and using Lemma \ref{lem_reg_m} (which yields that all the $M_{(\gamma, \lambda)}[g]$ are at least $C^3$ in $\varrho$ and $z$), we see that the regularity of $\mathfrak M_1[\zeta; \gamma, \lambda]$, $\mathfrak M_2[\zeta; \gamma, \lambda]$, and $\mathfrak M_5[\zeta; \gamma, \lambda]$ is at least that of $\xi$, i.e.~$C^{1,\alpha}(\mathbb R^2)$. In the formulas (\ref{m4_mg}) and (\ref{m6_mg})  for $\mathfrak M_6[\zeta; \gamma, \lambda]$ and $\mathfrak M_4[\zeta; \gamma, \lambda]$, respectively, we have the factors
\begin{align} 
&\frac{1}{\varrho} M_{(\gamma, \lambda)}[s(1 + \gamma \eta)](\varrho, \nu(\varrho, z), h(\varrho, z), A_t(\varrho, z), a(\varrho, z)), \label{already} \\
&\frac{1}{\varrho} e^{2(\xi-\gamma\nu)} M_{(\gamma, \lambda)}[s](\varrho, \nu, h, A_t, a). \label{new_sim}
\end{align} 
Since, as already observed, $M_{(\gamma, \lambda)}[g](\varrho, \nu(\varrho, z), h(\varrho, z), A_t(\varrho, z), a(\varrho, z))$ is odd in $\varrho$ if $g$ is odd in $s$ Lemma \ref{lem_axially_2} can be applied and this yields a regularity of $C^3$ in $\varrho$ and $z$, so in particular $C^{2,\alpha}(\mathbb R^2)$. \par
The term (\ref{already}) emerged already in the uncharged case treated in \cite{akr14}, the term (\ref{new_sim}) is new but similar. In the charged case, there appear some more problematic terms with factors $\varrho^{-1}$ in the source functions $\mathfrak g_4[\zeta; \gamma, \lambda]$ and $\mathfrak g_6[\zeta; \gamma, \lambda]$. Except for these problematic terms the source functions $\mathfrak g_i[\zeta; \gamma, \lambda]$ consist in products, sums, and compositions of functions which are at least $C^{1,\alpha}$ (namely the solution functions $\zeta$ and their derivatives which are chosen in $\mathcal X$). Consequently $\mathfrak g_1[\zeta; \gamma, \lambda] , \mathfrak g_2[\zeta; \gamma, \lambda], \mathfrak g_5[\zeta; \gamma, \lambda]$ are already in $C^{1,\alpha}(\mathbb R^2)$. It remains to consider the terms with $\varrho^{-1}$. These terms are 
\begin{equation} \label{crit_fun}
\frac{A_{t,\varrho} a}{\varrho}, \quad \frac{\nu_{,\varrho} a}{\varrho}, \quad \frac{h_{,\varrho} a}{\varrho}.
\end{equation}
We view $A_t$, $a$, $\nu$, and $h$ now as functions in $\varrho$, $z$ on $\mathbb R^2$ that are even in $\varrho$, cf.~Remark \ref{rem_even}. The functions $A_{t,\varrho} a$, $\nu_{,\varrho} a$, and $h_{,\varrho} a$ are odd in $\varrho$ and in $C^{2,\alpha}(\mathbb R^2)$, so in particular in $C^{2}(\mathbb R^2)$. So, by Lemma \ref{lem_axially_2}, the functions (\ref{crit_fun}) are in $C^1(\mathbb R^2)$ and consequently also in $C^{0,\alpha}(\mathbb R^2)$. This is sufficient to prove the asserted regularity. \par
Finally we consider the operator $\mathfrak G_3[\zeta; \gamma, \lambda]$. The asserted regularity is easy to see since the source function $\mathfrak g_3[\zeta; \gamma, \lambda]$ is obviously sufficiently regular, i.e.~$C^{0,\alpha}$.
\end{proof}

Next we check the decay properties of $\mathfrak G[\zeta; \gamma, \lambda]$. First we recall a technical lemma.

\begin{lemma} (Lemma 5.1 in \cite{akr14}) \label{lem_decay} \\
Let $f\in C^{0,\alpha}(\mathbb R^n)$, $n\geq 3$, fulfil $|f| \leq C (1+|x|)^{-(n+\epsilon)}$ for some constant $C>0$ and $\epsilon > 0$. Then $G_n[f] \in C^{2,\alpha}(\mathbb R^n)$, where $G_n[f]$ is defined in (\ref{def_greens_gn}), and there exists a constant $\tilde C>0$ such that for all multi indices $\sigma$, $|\sigma|\leq 2$, and for all $x\in\mathbb R^n$ we have
\begin{equation}
|\partial^\sigma G_n[f](x) | \leq \frac{\tilde C}{(1+|x|)^{n+|\sigma|-2}}.
\end{equation}
\end{lemma}


\begin{lemma} \label{lem_sol_decay}
Let $(\zeta; \gamma, \lambda)\in\mathcal U$. Then, there exists a constant $C>0$ such that for all $(\varrho,z)\in\mathbb R^2$, the following bounds hold:
\begin{align}
(\partial_\varrho + \partial_z) \mathfrak G_i[\zeta; \gamma, \lambda](\varrho,z) &\leq C \left(1+\sqrt{\varrho^2 + z^2}\right)^{-2}, \quad i=1,5, \\
(\partial_\varrho + \partial_z) \mathfrak G_2[\zeta; \gamma, \lambda](\varrho,z)  &\leq C \left(1+\sqrt{\varrho^2 + z^2}\right)^{-3}, \\
(\partial_\varrho + \partial_z) \mathfrak G_j[\zeta; \gamma, \lambda](\varrho,z) &\leq C \left(1+\sqrt{\varrho^2 + z^2}\right)^{-4}, \quad j = 4,6, \\
\mathfrak G_j[\zeta; \gamma, \lambda](\varrho,z) &\leq C \left(1+\sqrt{\varrho^2 + z^2}\right)^{-3}, \quad j = 4,6.
\end{align}
Furthermore the limits
\begin{equation} \label{def_lim_inf}
\nu_{\infty}^{\gamma, \lambda} := \lim_{|(\varrho,z)|\to\infty} \mathfrak G_1[\zeta; \gamma, \lambda](\varrho,z), \quad A_{\infty}^{\gamma, \lambda} := \lim_{|(\varrho,z)|\to\infty} \mathfrak G_5[\zeta; \gamma, \lambda](\varrho,z)
\end{equation}
exist.
\end{lemma}

\begin{proof}
By Lemma \ref{lem_decay} it suffices to check that the source functions $\mathfrak g_i[\zeta; \gamma, \lambda]$, $i=1,4,5,6$ and the matter functions $\mathfrak M_j[\zeta; \gamma, \lambda]$, $j=1,2,4,5,6$ have the right decay behaviour. In fact the matter functions do not have to be taken into account here, because they are of compact support, cf.~Lemma \ref{lem_comp_supp}. The source functions have to be investigated term by term. Since these terms consist in products of derivatives of the functions $\zeta_j$, $j=1,\dots,6$, it is easy to see that the necessary decay is available. \par
We illustrate this with the example of $\mathfrak g_1[\zeta; \gamma, \lambda]$. We have
\begin{align}
\mathfrak g_1[\zeta; \gamma, \lambda] &= -\frac{ h_{,\varrho}  \nu_{,\varrho} +  h_{,z}  \nu_{,z}}{1+ h} + \frac{\varrho^2}{2}(1+ h)^2 e^{-4\gamma \nu} \left( \omega_{,\varrho}^2 +  \omega_{,z}^2\right) \\
&\quad -\gamma^2 e^{-2\gamma\nu}  \left( ( A_{t,\varrho} + 2\omega\varrho  a + \omega \varrho^2  a_{,\varrho})^2 + ( A_{t,z} + \omega \varrho^2  a_{,z})^2 \right) \nonumber \\
&\quad - \gamma \frac{e^{2\gamma \nu}}{(1+ h)^2} \left((2 a+\varrho  a_{,\varrho})^2 + \varrho^2  a_{,z}^2 \right). \nonumber
\end{align}
We consider the first term $(h_{,\varrho} \nu_{,\varrho})/(1+h)$. Since $h\in \mathcal X_2$, $h > -1/2$ and $\nu \in \mathcal X_1$ we have
\begin{equation}
\frac{|h_{,\varrho}(\varrho, z) \nu_{,\varrho}(\varrho,z)|}{1+h} \leq 2\frac{ \left\| (1+|x|)^{3} \nabla h \right\|_\infty \left\| (1+|x|)^{1+\beta} \nabla \nu \right\|_\infty}{ \left(1+|x|\right)^{4+\beta}}\leq \frac{C}{\left(1+|x|\right)^{4+\beta}}.
\end{equation}
The remaining terms are treated in a similar fashion. \par
Finally, by inspecting the formula (\ref{def_g15}) for the solution operators $\mathfrak G_1$ and $\mathfrak G_5$ corresponding to $\nu$ and $A_t$, respectively, we see that 
\begin{align*}
&\mathfrak G_1[\zeta; \gamma, \lambda](\varrho,z) + \frac{1}{|\mathbb S^2|} \int_{\mathbb R^3} \frac{\mathfrak M_1[\zeta; \gamma, \lambda](\varrho_y, z_y)}{|y|} \mathrm dy, \\
&\mathfrak G_5[\zeta; \gamma, \lambda](\varrho,z) + \frac{1}{|\mathbb S^2|} \int_{\mathbb R^3} \frac{\mathfrak M_5[\zeta; \gamma, \lambda](\varrho_y, z_y)}{|y|} \mathrm dy
\end{align*}
decay towards spatial infinity, also by Lemma \ref{lem_decay}.
\end{proof}

\begin{rem}
Note that in Lemma \ref{lem_decay} we have seen that for the functions $\nu$ and $A_t$ the decay is improved, form $(1+|x|)^{-(1+\beta)}$ to $(1+|x|)^{-2}$, i.e.~assuming the weaker decay of $\nu, A_t \in \mathcal X_1$ we obtain the stronger decay of $\mathfrak G_1[\zeta; \gamma, \lambda]$, $\mathfrak G_5[\zeta; \gamma, \lambda]$. This is important in the proof that the Fr\'echet derivative of these components at $(\zeta_0; 0,0)$ is a compact operator, which in turn plays a role in the proof that this derivative is a bijection, cf.~Lemma \ref{lem_bijection} below and \cite[Lemma 6.2]{akr11}.
\end{rem}
\noindent All required properties of $\mathfrak G[\zeta; \gamma, \lambda]$ are now verified, thus the operator $\mathfrak F$ is well defined.

\section{The Fr\'echet derivative of $\mathfrak F$} \label{sect_frechet_dir}

We denote the functions $\nu$, $h$, $\xi$, $\omega$, $A_t$, $a$ constituting the collection $\zeta$ by $\zeta_1, \dots, \zeta_6$, if convenient. The Fr\'echet derivative of $\mathfrak G_i$ with respect to $\zeta_j$ at $(\zeta; \gamma, \lambda)$ is a linear operator from $\mathcal X_j$ to $\mathcal X_i$, $i,j=1,\dots,6$. Here and in the remainder of the article by $\mathcal X_5$ and $\mathcal X_6$ we mean $\mathcal X_1$ and $\mathcal X_4$, respectively, since these are the function spaces corresponding to $\zeta_5$ and $\zeta_6$, respectively. We denote the Fr\'echet derivative by
\begin{equation}
D_{\zeta_j} \mathfrak G_i[\zeta; \gamma, \lambda] : \mathcal X_j \to \mathcal X_i, \quad \delta \zeta_j \mapsto \left(D_{\zeta_j} \mathfrak G_i[\zeta; \gamma, \lambda]\right)\delta \zeta_j.
\end{equation}

\begin{prop}
The operators $\mathfrak G_i:\mathcal U \to \mathcal X_i$, $i=1,\dots,6$ are continuous and continuously Fr\'echet differentiable with respect to $\nu, \xi, h, \omega, A_t, a$.
\end{prop}

\begin{proof}
The operators $\mathfrak G_i$, $i=1,2,4,5,6$ are of similar structure and we will start by analysing these operators. Schematically one can write these operators as sums of expressions of the form
\begin{equation}
\mathfrak G_\Phi[\zeta; \gamma,\lambda](\varrho,z) = \int_{\mathbb R^n} G^n_y(|\varrho|,0,\dots, 0, z) \, \Phi^{(\gamma, \lambda)}(\varrho(y), \zeta(y), \zeta_{,\varrho}(y), \zeta_{,z}(y)) \, \mathrm dy
\end{equation}
where the function $\Phi^{(\gamma, \lambda)}:\mathbb R^{19}\to\mathbb R$ is a placeholder for either $g_i^{(\gamma)}$ or $M_i^{(\gamma, \lambda)}$. In order to write this in a compact and handy way we define the functional $\tilde G_n$ (which is slightly different from $G_n$, cf.~the definition (\ref{def_greens_gn}) of $G_n$) by
\begin{equation}
\tilde G_n\!\left[\Phi^{(\gamma,\lambda)},\zeta \right]\!(\varrho, z) := \int_{\mathbb R^n} G^n_y(|\varrho|,0,\dots, 0, z) \, \Phi^{(\gamma, \lambda)}(\varrho(y), \zeta(y), \zeta_{,\varrho}(y), \zeta_{,z}(y)) \, \mathrm dy.
\end{equation}
We will check now that the Fr\'echet derivative of $\mathfrak G_\Phi$ with respect to $\zeta_j$ is given by
\begin{multline} \label{f_der_dphi}
\left(D_{\zeta_j} \mathfrak G_\Phi[\zeta; \gamma, \lambda] \delta \zeta_j\right)(\varrho, z) \\= \tilde G_n\!\left[\left(\partial_{\zeta_j}\Phi^{(\gamma, \lambda)} \delta \zeta_j\right) + \left(\partial_{\zeta_{j,\varrho}}\Phi^{(\gamma, \lambda)} \, \partial_\varrho\left(\delta \zeta_{j}\right) \right) + \left(\partial_{\zeta_{j,z}} \Phi^{(\gamma, \lambda)} \, \partial_z\left(\delta \zeta_{j}\right)\right), \zeta\right].
\end{multline}
So we have to check that 
\begin{align}
&\Big\| \tilde G_n[\Phi^{(\gamma, \lambda)},\zeta+\delta \zeta_j] - \tilde  G_n[\Phi^{(\gamma, \lambda)},\zeta] \\
&\quad - \tilde G_n\!\left[\left(\partial_{\zeta_j}\Phi^{(\gamma, \lambda)} \delta \zeta_j\right) + \left(\partial_{\zeta_{j,\varrho}}\Phi^{(\gamma, \lambda)} \, \partial_\varrho\left(\delta \zeta_{j}\right) \right) + \left(\partial_{\zeta_{j,z}} \Phi^{(\gamma, \lambda)} \, \partial_z\left(\delta \zeta_{j}\right)\right), \zeta\right]\Big\|_{\mathcal X_\Phi} \nonumber \\
&= o\left(\|\delta \zeta_j\|_{\mathcal X_\Phi}\right). \nonumber 
\end{align}
Here $\mathcal X_\Phi$ is the function space corresponding to $\Phi^{(\gamma,\lambda)}$. I.e.~if $\Phi^{(\gamma,\lambda)}$ is for example $M_1^{(\gamma,\lambda)}$ then $\mathcal X_\Phi$ is $\mathcal X_1$. Define $m$ as the number how often functions in $\mathcal X_\Phi$ are continuously differentiable, i.e.~the largest number such that $\mathcal X_\Phi \subset C^{m,\alpha}$. By the standard elliptic estimate \cite[Theorem 10.3]{ll00} and the inclusion $C^{m+1} \subset C^{m,\alpha}$ it suffices to check
\begin{align}
&\sum_{|\sigma|\leq {m-1}} \bigg\| \partial^\sigma \bigg( 
\Phi^{(\gamma, \lambda)}(\cdot, \zeta + \delta \zeta_j, \nabla (\zeta+\delta \zeta_j)) 
- \Phi^{(\gamma, \lambda)}(\cdot, \zeta, \nabla \zeta) \label{cond_diffbar} \\
& \hspace{2.5cm} - \partial_{\zeta_j}\Phi^{(\gamma, \lambda)}(\cdot,\zeta,\nabla\zeta) \delta \zeta_j 
- \partial_{\zeta_{j,\varrho}}\Phi^{(\gamma, \lambda)}(\cdot,\zeta,\nabla\zeta) \, \partial_\varrho\left(\delta \zeta_{j}\right) \nonumber \\
& \hspace{6.1cm} - \partial_{\zeta_{j,z}} \Phi^{(\gamma, \lambda)}(\cdot,\zeta,\nabla\zeta) \, \partial_z \left(\delta \zeta_{j}\right)
\bigg) \bigg\|_\infty \leq o( \| \delta \zeta_j \|_{\mathcal X_i}). \nonumber
\end{align}
It turns out that (\ref{cond_diffbar}) holds if the functions $\Phi^{(\gamma, \lambda)}$ are sufficiently regular, i.e.~in $C^{m}$ to be precise. Now, $\Phi^{(\gamma, \lambda)}$ is either a source function $g_i^{(\gamma)}$ or a matter function $M_i^{(\gamma,\lambda)}$. The source functions are smooth in all of the variables $\zeta$, $\zeta_{,\varrho}$, and $\zeta_{,z}$, since they involve only the exponential function and addition, multiplication and division by $1+h$. Note here that $1+h > \frac 12$ if $(\zeta; \gamma, \lambda)\in\mathcal U$. \par
For the matter functions $M_i^{(\gamma, \lambda)}$, $i=1,2,4,5,6$, defined in equations (\ref{def_hat_m_1})--(\ref{def_hat_m_6}), we note that they do not depend on derivatives of $\zeta$ and that the regularity is determined by the functions $M_{(\gamma,\lambda)}[g]$ which are all $C^3$ by Lemma \ref{lem_reg_m} and this is sufficient. \par
The operator $\mathfrak G_3$ is easier to treat since the expression (\ref{def_g3}) can be expanded explicitly in powers of $\delta h$, $\delta \nu$, $\delta \omega$, $\delta A_t$, and $\delta a$. Note again that $1+h$ is bounded away from zero for all $(\zeta; \gamma, \lambda)\in \mathcal U$.
\end{proof}

In the next step we calculate the Fr\'echet derivatives of $\mathfrak G_i$, $i=1,\dots,6$ and evaluate them at $(\zeta_0;0,0)$. The parts of $\mathfrak G_i$, $i = 1,\dots,6$ involving the source functions $\mathfrak g_i$ can be expanded directly, i.e.~we calculate the Fr\'echet derivative at $(\zeta_0; 0, 0)$ by replacing $\mathfrak g_i[\zeta;\gamma](\varrho, z)$ in the integral expressions (\ref{def_g15})--(\ref{def_g46}) with the $\epsilon$-derivatives of $\mathfrak g_i[\zeta + \epsilon \delta\zeta_j;\gamma](\varrho, z)$ evaluated at $\epsilon = 0$ and then at $(\zeta_0;0,0)$. The non-zero derivatives are
\begin{align}
\left[\partial_\epsilon \mathfrak g_1[\zeta + \epsilon \delta h;\gamma](\varrho, z) \Big|_{\epsilon = 0}\right]_{(\zeta; \gamma, \lambda)=(\zeta_0;0,0)} &= - (\nabla U_N \cdot \nabla \delta h)(\varrho, z), \\
\left[\partial_\epsilon \mathfrak g_3[\zeta + \epsilon \delta h;\gamma](\varrho, z) \Big|_{\epsilon = 0}\right]_{(\zeta; \gamma, \lambda)=(\zeta_0;0,0)} &= \frac{\varrho}{2} (\partial_{\varrho\varrho} \delta h - \partial_{zz} \delta h)( \varrho, z) + \partial_\varrho \delta h(\varrho, z), \\
\left[\partial_\epsilon \mathfrak g_5[\zeta + \epsilon \delta h;\gamma](\varrho, z) \Big|_{\epsilon = 0}\right]_{(\zeta; \gamma, \lambda)=(\zeta_0;0,0)} &= - (\nabla A_N \cdot \nabla \delta h)(\varrho, z).
\end{align}
The notation here should be interpreted as $\zeta + \epsilon \delta h = (\nu, h + \epsilon \delta h, \xi, \omega, A_t, a)$. For the parts involving the matter functions we use formula (\ref{f_der_dphi}), where $\Phi^{(\gamma, \lambda)}$ is replaced by the matter functions $M_i^{(\gamma, \lambda)}$, $i=1,\dots,6$, given in  (\ref{def_hat_m_1})--(\ref{def_hat_m_6}). The matter functions  $M_i^{(\gamma, \lambda)}$, $i=1,\dots,6$ depend only on $\zeta$ and not on its derivatives. \par
First we consider the matter functions 
\begin{align*}
M_1^{(\gamma,\lambda)}(\varrho, \zeta) &:= 4\pi e^{2(\xi-\gamma\nu)} \int_{\mathbb R_v^3} \phi(E)\psi(\lambda, L) \frac{1 + 2\gamma |v|^2}{\sqrt{1+\gamma |v|^2}} \, \mathrm d^3v, \\
M_2^{(\gamma,\lambda)}(\varrho, \zeta) &:= 8\pi \gamma^2 (1+h) e^{2(\xi-\gamma\nu)} \int_{\mathbb R_v^3} \phi(E)\psi(\lambda, L) \frac{(v^1)^2+(v^2)^2}{\sqrt{1+\gamma |v|^2}} \, \mathrm d^3 v,\\
M_4^{(\gamma,\lambda)}(\varrho, \zeta) &:= -\frac{16\pi \gamma}{\varrho(1+h)} e^{2\xi}\int_{\mathbb R_v^3} \phi(E)\psi(\lambda, L) v^3\,\mathrm d^3 v, \\
\end{align*}
of the Einstein equations, given in (\ref{def_hat_m_1})--(\ref{def_hat_m_4}), where $\mathrm d^3v = \mathrm dv^1\mathrm dv^2 \mathrm dv^3$. If one calculates the derivative of $M_i^{(\gamma, \lambda)}(\varrho, \zeta)$, $i=1,2,4$, with respect to any of the arguments $\nu, h, \xi, \omega, A_t, a$ one obtains back an expression with the same structure, possibly with the function $\partial_{\zeta_j}(\phi(E)\psi(\lambda,L))$ instead of $\phi(E)\psi(\lambda,L)$ in the integral. \par
In the limit $\gamma\to 0$ only the terms where the $\gamma$-factors cancel will remain. Thus $\partial_{\zeta_j}M_i^{(\gamma, \lambda)}(\varrho, \zeta) |_{\gamma=0} = 0$ for $i=2,4$, and 
\begin{align}
\partial_{\xi} M_1^{(\gamma, \lambda)}(\varrho, \zeta) \Big|_{\gamma=0} &= 2 \mathfrak M_1[\zeta; \gamma, \lambda] \Big|_{\gamma=0}, \\
\partial_{\zeta_j} M_1^{(\gamma, \lambda)}(\varrho, \zeta) \Big|_{\gamma=0} &= 4\pi  e^{2(\xi - \gamma\nu)} \int_{\mathbb R_v^3} \partial_{\zeta_j} \left( \phi(E) \psi(\lambda,L) \right) \sqrt{1 + \gamma |v|^2}\, \mathrm d^3 v\Big|_{\gamma=0}, \label{second_d_m}
\end{align} 
where $j=1,2,4,5,6$. Consider now the term $\partial_{\zeta_j}(\phi(E)\psi(\lambda,L))$ in (\ref{second_d_m}). First we observe that the assumption $\psi(0,L) = 1$ implies $\partial_L \psi(0,L) = 0$. This yields already
\begin{align}
\lim_{(\gamma,\lambda) \to (0,0)} \left[\partial_{\zeta_j}(\phi(E) \psi(\lambda, L))\right]_{\zeta = \zeta_0} &= \lim_{(\gamma,\lambda) \to (0,0)} \left[\psi(\lambda, L) \partial_{\zeta_j}\phi(E) \right]_{\zeta = \zeta_0} \\
&= \lim_{(\gamma,\lambda) \to (0,0)} \left[\partial_{\zeta_j}\phi(E) \right]_{\zeta = \zeta_0}. \label{ext_psi_ass}
\end{align}
If we now set $\zeta=\zeta_0$ and consider the limit $(\gamma, \lambda) \to (0, 0)$ only the derivatives with respect to $\nu$ and $A_t$ are non-vanishing. The derivative with respect to $a$ vanishes due to (\ref{ext_psi_ass}) and the fact that $E$ is independent of $a$. The derivatives with respect to $h$, $\xi$, $\omega$ vanish by symmetry. This can be seen as follows. We have 
\begin{equation}
E|_{\zeta = \zeta_0} =  \frac{e^{\gamma \nu_N} \sqrt{1+\gamma |v|^2}- 1}{\gamma} + q A_N, 
\end{equation}
and therefore
\begin{equation} \label{lim_p_p}
\lim_{(\gamma, \lambda) \to (0, 0)} \phi \left(E|_{\zeta = \zeta_0}\right) \psi \left( \lambda, L|_{\zeta = \zeta_0} \right) = \phi\left(\frac{|v|^2}{2} + \nu_N + q A_N\right)
\end{equation}
where the Newtonian limit (\ref{newtonian_l_e}) of the energy and the assumption $\psi(0, L)=1$ on $\psi$ has been used. Observe that the limit (\ref{lim_p_p}) is even in $v^1$, $v^2$, $v^3$. Consider next the derivatives
\begin{equation}
\partial_h E = \varrho \omega e^{-\gamma\nu} v^3, \quad \partial_\xi E = 0, \quad \partial_\omega E = \varrho (1+h) e^{-\gamma\nu} v^3.
\end{equation}
These derivatives are either zero or odd in $v^3$. Integration over an odd-in-$v^3$ function yields zero. \par
For the derivatives with respect to $\nu$ and $A_t$ the same principles apply, however not all terms vanish. Consider for example $\partial_\nu M_1^{(\gamma, \lambda)}(\varrho, \zeta)$. One obtains
\begin{align*}
\partial_\nu M_1^{(\gamma,\lambda)}(\varrho, \zeta) &= -8\pi \gamma e^{2(\xi-\gamma\nu)} \int_{\mathbb R_v^3} \phi(E)\psi(\lambda, L) \frac{1 + 2\gamma |v|^2}{\sqrt{1+\gamma |v|^2}}\, \mathrm d^3v \\
& \quad + 4\pi e^{2(\xi-\gamma\nu)} \int_{\mathbb R_v^3} \phi(E)' \psi(\lambda, L) \left(e^{\gamma\nu} + \frac{\gamma \omega \varrho (1+h) v^3}{\sqrt{1+\gamma |v|^2}} \right) \left(1 + 2\gamma |v|^2\right)\, \mathrm d^3v \\
& \quad - 4\pi \gamma \varrho (1+h) e^{2\xi - 3\gamma\nu} \int_{\mathbb R_v^3} \partial_\nu \phi(E) \partial_L\psi(\lambda, L) v^3 \frac{1 + 2\gamma |v|^2}{\sqrt{1+\gamma |v|^2}}\, \mathrm d^3v.
\end{align*}
So the derivatives of the matter function $M_1^{(\gamma, \lambda)}$ of the Einstein equations which are non-vanishing at $(\zeta_0; 0, 0)$ are
\begin{align}
\partial_{\nu} M_1^{(\gamma, \lambda)} \Big|_{(\zeta_0; 0,0)} &= 4\pi \alpha_N, \\
\partial_{\xi} M_1^{(\gamma, \lambda)} \Big|_{(\zeta_0; 0,0)} &= 8\pi \rho_N, \\
\partial_{A_t} M_1^{(\gamma, \lambda)} \Big|_{(\zeta_0; 0,0)} &= 4\pi q \alpha_N,
\end{align} 
where $\rho_N$ and $\alpha_N$ are defined in (\ref{def_rho_u}) and (\ref{def_alpha_u}), respectively. Next we consider the matter functions 
\begin{align*}
M_5^{(\gamma,\lambda)}(\varrho, \zeta) &:= 4\pi q e^{2\xi-3\gamma\nu} \int_{\mathbb R_v^3} \phi(E)\psi(\lambda, L) \left(e^{2\gamma\nu} + \frac{\gamma \varrho (1+h) \omega v^3}{\sqrt{1+\gamma |v|^2}}\right) \, \mathrm d^3 v, \\
M_6^{(\gamma,\lambda)}(\varrho, \zeta) &:= -4\pi q \gamma \varrho (1+h) e^{2\xi-3\gamma\nu} \int_{\mathbb R_v^3} \phi(E)\psi(\lambda, L) \frac{v^3}{\sqrt{1+\gamma |v|^2}}\, \mathrm d^3 v,
\end{align*}
of the Maxwell equations in the representations given in (\ref{def_hat_m_5})--(\ref{def_hat_m_6}). The first observation is that if $\gamma=0$ then all terms but the first one of $M_5^{(\gamma, \lambda)}$ vanish. So we only need to discuss the derivatives of
\begin{equation}
4\pi q e^{2\xi-\gamma\nu} \int_{\mathbb R_v^3} \phi(E) \psi(\lambda,L) \, \mathrm d^3 v.
\end{equation}
By the same reasoning as above we obtain
\begin{align}
\partial_{\nu} M_5^{(\gamma, \lambda)} \big|_{(\zeta_0; 0,0)} &= 4 \pi q \alpha_N, \\
\partial_{\xi} M_5^{(\gamma, \lambda)}  \big|_{(\zeta_0; 0,0)} &= 8 \pi q \rho_N, \\
\partial_{A_t} M_5^{(\gamma, \lambda)} \big|_{(\zeta_0; 0,0)} &= 4 \pi q^2 \alpha_N.
\end{align}
We denote the Fr\'echet derivative of $\mathfrak F$ with respect to $\zeta$, at $(\zeta_0;0,0)$, by $\mathfrak L$, i.e.
\begin{align}
&\mathfrak L := D\mathfrak F[\zeta_0; 0,0]: \mathcal X \to \mathcal X, \label{def_l} \\
&\delta \zeta \mapsto \mathfrak L(\delta \zeta) = (\delta \nu - \mathfrak L_1(\delta \nu, \delta h, \delta \xi, \delta A_t), \delta h, \delta \xi - \mathfrak L_3(\delta h), \delta \omega, \delta A_t - \mathfrak L_5(\delta\nu, \delta \xi, \delta A_t), \delta a), \nonumber
\end{align}
where
\begin{align}
\mathfrak L_1(\delta \nu, \delta h, \delta \xi, \delta A_t) &= -\mathfrak L_1^{(1)}(\delta \nu + q \delta A_t) - \mathfrak L_1^{(2)} (\delta \xi) + \mathfrak L_1^{(3)}(\delta h), \\
\mathfrak L_3(\delta h) &= \delta h(0,z) + \int_0^\varrho \left(\frac{s}{2} (\partial_{\varrho\varrho} \delta h - \partial_{zz} \delta h)(s,z) + \partial_\varrho \delta h(s, z) \right) \mathrm ds,\\
\mathfrak L_5(\delta \nu, \delta h, \delta \xi, \delta A_t) &= q \mathfrak L_1^{(1)}(\delta \nu + q \delta A_t) + q \mathfrak L_1^{(2)} (\delta \xi) + q \mathfrak L_5^{(3)}(\delta h),
\end{align}
where
\begin{align}
\mathfrak L_1^{(1)}(\delta u) &= \int_{\mathbb R^3} \left(\frac{1}{|x-y|} - \frac{1}{|y|} \right) \alpha_N(|y|) \delta u(\varrho_y, z_y) \, \mathrm dy, \\
\mathfrak L_1^{(2)}(\delta \xi) &=  2 \int_{\mathbb R^3}  \left(\frac{1}{|x-y|} - \frac{1}{|y|} \right) \rho_N(|y|) \delta \xi(\varrho_y, z_y) \, \mathrm dy, \\
\mathfrak L_1^{(3)}(\delta h) &=  \frac{1}{4\pi} \int_{\mathbb R^3} \frac{1}{|x-y|} \nabla \nu_N(|y|) \cdot \nabla (\delta h)(\varrho_y, z_y) \, \mathrm dy, \\
\mathfrak L_5^{(3)}(\delta h) &=  \frac{1}{4\pi} \int_{\mathbb R^3} \frac{1}{|x-y|} \nabla A_N(|y|) \cdot \nabla (\delta h)(\varrho_y, z_y) \, \mathrm dy.
\end{align}

The shorthands $\rho_N = \rho_{U_N}$ and $\alpha_N = \alpha_{U_N}$ are defined in (\ref{def_rho_u}) and (\ref{def_alpha_u}), respectively, where $U_N = \nu_N + q A_N$, and the functions $\nu_N$ and $A_N$ are defined as the solutions of the system (\ref{newt_1})--(\ref{newt_2}). 

\begin{lemma} \label{lem_bijection}
$\mathfrak L$ is a bijection.
\end{lemma}

\begin{proof}
First we prove that $\mathfrak L$ is injective. Since $\mathfrak L$ is linear it suffices to show that $\mathrm{ker}(\mathfrak L) = 0$. Let $\delta \zeta \in \mathcal X$ such that $\mathfrak L(\delta \zeta) = 0$. From the definition of $\mathfrak L$ in (\ref{def_l}) we immediately read off $\delta h = \delta \omega = \delta a = 0$. Consequently $\mathfrak L_3(\delta h)=0$ and therefore also $\delta \xi = 0.$ Since $\delta h = \delta\xi = 0$ and thus $\mathfrak L_1^{(3)}(\delta h) = \mathfrak L_5^{(3)}(\delta h) = \mathfrak L_1^{(2)}(\delta \xi) =0$ we can furthermore read off $\delta A_t = -q \delta \nu$. We finish the proof of injectivity by showing that $\delta \nu + q \delta A_t = 0$. To simplify notation we denote in the following $\delta u = \delta\nu + q\delta A_t \in \mathcal X_1$. Those two identities will then imply $(1-q^2)\delta u=0$ and therefore $\delta \nu=0$ and $\delta A_t = 0$. \par
Adding the first and $q$ times the fifth component of $\mathfrak L(\delta\zeta) = 0$ yields
\begin{equation}
\delta u = - \left(1-q^2\right) \int_{\mathbb R^3} \left(\frac{1}{|x-y|} - \frac{1}{|y|} \right) \alpha_N(|y|) \, \delta u(\varrho_y, z_y) \, \mathrm dy.
\end{equation}
This is a solution of
\begin{align}
\Delta (\delta u) &= \left( 1 - q^2 \right)  \alpha_N \, \delta u, \label{zero_p_eq_1} \\
(\delta u)(0) &= 0. \label{zero_p_eq_2}
\end{align}
In \cite[Section 6]{akr11} it has been shown that this is the only solution of (\ref{zero_p_eq_1})--(\ref{zero_p_eq_2}), provided that $6+4\pi r^2 (1-q^2) \alpha_N(r) > 0$ which is assumed. \par
Next we show that $\mathfrak L$ is surjective. Let $b=(b_1,\dots, b_6)\in \mathcal X$ be given. The aim is now to construct $\delta \zeta = (\delta \nu, \delta h, \delta \xi, \delta\omega, \delta A_t, \delta a)\in\mathcal X$ such that 
\begin{equation} \label{ldzg}
\mathcal L(\delta \zeta) = b.
\end{equation}
By inspecting the formula (\ref{def_l}) of $\mathfrak L$ we immediately see that we have to choose $\delta h = b_2$, $\delta \omega = b_4$, $\delta a = b_6$. In the third component of (\ref{ldzg}) we obtain
\begin{equation}
\delta \xi =  b_3 + \mathfrak L_3(\delta h),
\end{equation}
which is in $\mathcal X_3$ since $\mathfrak L_3(\delta h) \in \mathcal X_3$. It remains to construct $\delta \nu$ and $\delta A_t$. Note first that $\mathfrak L_1^{(2)}(\delta\xi), \mathfrak L_1^{(3)}(\delta h), \mathfrak L_5^{(3)}(\delta h) \in \mathcal X_1$ (recall $\mathcal X_5 = \mathcal X_1$). We add the first component of (\ref{ldzg}) and $q$ times the fifth component of (\ref{ldzg}). We obtain
\begin{equation}
\delta u - \left(1-q^2\right) \mathfrak L_1^{(1)}(\delta u) = (b_1 + q b_5) - \left(1-q^2\right) \mathfrak L_1^{(2)}(\delta \xi) + \left(\mathfrak L_1^{(3)} + q^2 \mathfrak L_5^{(3)}\right) (\delta h). 
\end{equation}
This equation has a solution $\delta u \in \mathcal X_1$ since the operator $\mathfrak L_1^{(1)}$ is compact. This has been established in \cite[Lemma 6.2]{akr11}. Then, considering the first component of (\ref{ldzg}) again, we can construct $\delta \nu$ via
\begin{equation}
\delta \nu = b_1 - \mathfrak L_1^{(1)}(\delta u) - \mathfrak L_1^{(2)}(\delta \xi) +\mathfrak L_1^{(3)}(\delta h).
\end{equation}
Finally, we obtain $\delta A_t$ via $\delta A_t = \frac{1}{q} (\delta u-\delta\nu)$.
\end{proof}

\section{Application of the implicit function theorem}

In the preceding sections we have established that the solution operator $\mathfrak F$ fulfils the assumptions of the implicit function theorem for Banach spaces. Now we can prove the following proposition.

\begin{prop} \label{prop_appli}
There exist solutions $\zeta = (\nu, h, \xi, \omega, A_t, a)$ to the reduced EVM-system (\ref{final_eq_nu})--(\ref{final_eq_a}) with parameters $\gamma \in [0,\delta)$, $\lambda \in (-\delta, \delta)$ if $\delta$ is chosen sufficiently small that satisfy the boundary conditions (\ref{bc_infinity}) and (\ref{bc_center}).
\end{prop}

\begin{proof}
The solution $\zeta  = (\nu, h, \xi, \omega, A_t, a)$ exists by virtue of the implicit function theorem. The functions $\omega$, $\xi$, $h$, and $a$ fulfil the boundary condition 
\begin{equation*}
\lim_{|(\varrho, z)\to \infty} (|\omega| + |\xi| + |h| + |a|) = 0
\end{equation*}
by construction. For $\omega$ and $a$ see the definition (\ref{def_norm_x4}) of the norm of the space $\mathcal X_4$. Analogously, with Lemma \ref{lem_decay}, it follows that $h$ fulfils the boundary condition. By inspecting the structure (\ref{def_g3}) of the solution operator $\mathfrak G_3$ one easily sees that the boundary condition 
\begin{equation} \label{bc_center_2}
\xi(0,z) = \ln(1+h(0,z))
\end{equation}
is satisfied, too. For the boundary condition of $\xi$ at infinity one infers first from (\ref{bc_center_2}) that $\lim_{|z|\to\infty}\xi(0,z) = 0$, and then the decay as $\varrho\to\infty$ can be deduced from the decay of the integrand of the solution operator $\mathfrak G_3$, cf.~formula (\ref{def_g3}) and \cite[Prop. 2.3]{akr11}. The solution functions $\nu$, $A_t$ obtained from the implicit function theorem do however a priori not satisfy the boundary condition
\begin{equation*}
\lim_{|(\varrho, z)| \to \infty} (|\nu| + A_t) = 0
\end{equation*}
and we define
\begin{equation*}
 \nu_\infty^{\gamma, \lambda} := \lim_{|(\varrho, z)| \to \infty} |\nu|, \quad A_\infty^{\gamma, \lambda} := \lim_{|(\varrho, z)| \to \infty} |A_t|.
\end{equation*}
A rescaling is necessary. The functions 
\begin{equation*}
\nu - \nu_\infty^{\gamma, \; \lambda}, \; \mu + \gamma \nu_\infty^{\gamma, \lambda}, \; h, \; e^{-\gamma \nu_\infty^{\gamma, \lambda}}, \; e^{-\gamma\nu_\infty^{\gamma, \lambda}} (A_t-A_\infty^{\gamma,\lambda}), \; e^{\gamma \nu_\infty^{\gamma, \lambda}} a
\end{equation*}
then fulfil the reduced EVM-system with the boundary conditions which correspond to an asymptotically flat solution.
\end{proof}

\section*{Appendix}

\begin{proof}[Proof of Lemma \ref{lem_conserved_quan}]
Recall the definition of the transport operator,
\begin{equation*}
\mathfrak T = p^\mu \partial_\mu + \left(q F^\gamma{}_{\mu} \, p^\mu - \Gamma^\gamma_{\alpha\beta} p^\alpha p^\beta\right) \partial_{p^\gamma}.
\end{equation*}
Now it shall be expressed with respect to the frame coordinates (\ref{def_frame}). First we derive the form of $\mathfrak T$ with respect to a general orthonormal frame $e_a = e_a{}^\alpha \partial_{x^\alpha}$ (and corresponding co-frame $\alpha^b = e^b{}_\beta \mathrm dx^\beta$). Using the definitions
\begin{equation}
\Gamma_{\alpha\beta}^\gamma = \mathrm dx^\gamma \left(\nabla_\beta \partial_\alpha \right), \quad \Gamma_{ab}^c = \alpha^c \left(\nabla_{e_b} e_a \right)
\end{equation}
for $\Gamma_{\alpha\beta}^\gamma$ and $\Gamma_{ab}^c$ one derives the transformation law
\begin{equation} \label{trafo_christoffels}
\Gamma_{ab}^c = e^c{}_\alpha e_b{}^\beta \partial_\beta e_a{}^\alpha + e^c{}_\gamma e_b{}^\beta e_a{}^\alpha \Gamma_{\alpha\beta}
^\gamma.
\end{equation}
Furthermore the change of variables
\begin{equation}
x^\mu \mapsto y^\mu = x^\mu, \quad p^\nu \mapsto v^a = e^a{}_\nu p^\nu
\end{equation}
entails the replacements
\begin{equation} \label{trafo_derivatives}
\partial_{x^\mu} = \partial_{y^\mu} + e_b{}^\alpha v^b \partial_\mu e^a{}_\alpha \partial_{v^a}, \quad  \partial_{p^\nu} = e^a{}_\nu \partial_{v^a}.
\end{equation}
This yields
\begin{equation}
\mathfrak T = v^a e_a{}^\alpha \partial_\alpha + \left(q F^c{}_a v^a- \Gamma_{ab}^c v^a v^b \right) \partial_{v^c}.
\end{equation}
In order to obtain the explicit expression for the transport operator $\mathfrak T$ with respect to the frame coordinates (\ref{the_coords}) we apply the transformation laws (\ref{trafo_christoffels}) and (\ref{trafo_derivatives}) to the frame (\ref{frame_matrix}), where the Christoffel symbols
\begin{equation}
\Gamma_{\alpha\beta}^\gamma = \frac 12 g^{\gamma\delta} \left( \partial_\alpha g_{\beta\delta} + \partial_\beta g_{\delta\alpha} - \partial_\delta g_{\alpha\beta} \right)
\end{equation}
are calculated from the ansatz (\ref{ansatz_metric}) for the metric. The transport operator is then explicitly given by
\begin{align*}
\mathfrak T &= v^0 e^{-\gamma\nu} \partial_t + e^{-\mu} (v^1 \partial_\varrho + v^2 \partial_z) + \left(v^0 e^{-\gamma\nu} \omega + v^3 \frac{e^{\gamma\nu}}{\varrho H}\right) \partial_\varphi \\
&\quad -q e^{-\mu - \frac{\nu}{c^2}} \left(\left( A_{t,\varrho} + \omega A_{\varphi,\varrho} \right) \Omega_{01}^V + \left(A_{t,z} + \omega A_{\varphi,z}\right) \Omega_{02}^V \right) \\
&\quad + \frac{q}{\varrho H} e^{-\mu + \frac{\nu}{c^2}} \left(A_{\varphi,\varrho}\Omega_{13}^V + A_ {\varphi,z} \Omega_{23}^V \right) + q e^{-2\mu} \left(A_{\varrho,z} - A_{z,\varrho}\right) \Omega_{21}^V \\
&\quad + e^{-\mu} \frac{v^3}{c^2} \left(\nu_{,\varrho} \Omega_{13}^V + \nu_{,z} \Omega_{23}^V \right) - e^{-\mu} v^0 \left(\nu_{,\varrho} \Omega_{01}^V + \nu_{,z} \Omega_{02}^V \right) + e^{-\mu} \left(v^2 \mu_{,\varrho} - v^1 \mu_{,z}\right) \Omega_{21}^V \\
&\quad+ e^{-\mu} \frac{v^3}{H} \left(H_{,\varrho} \Omega_{31}^V + H_{,z} \Omega_{32}^V \right) + \frac{v^3}{\varrho} e^{-\mu} \Omega_{31}^V - e^{-\mu - 2\frac{\nu}{c^2}} \varrho H v^3 \left(\omega_{,\varrho}\Omega_{01}^V + \omega_{,z} \Omega_{02}^V \right)
\end{align*}
where we use the shorthands
\begin{align}
\Omega_{ij}^V := v^i \partial_{v^j} - v^j \partial_{v^i}, \quad \Omega_{0i}^V := \frac{v^i}{c^2} \partial_{v^0} + v^0 \partial_{v^i}.
\end{align}
Now, the transport operator can be applied to the quantities
\begin{align*}
L &= \varrho H e^{ -\gamma\nu} v^3 - q A_\varphi, \\
E &= \frac{e^{\gamma\nu} v^0 - 1}{\gamma} + \omega \varrho H e^{ -\gamma\nu} v^3 + q A_t,
\end{align*}
where we note that $E$ only depends on the variables $\varrho$, $z$, $v^0$, and $v^3$, and $L$ only depends on the variables $\varrho$, $z$, and $v^3$.
\end{proof}

\end{document}